\documentclass[11pt,a4paper]{amsart}

\usepackage{cmap}                
\usepackage[english]{babel}
\usepackage{microtype}

\usepackage{amssymb}
\usepackage{mathtools}           
\usepackage{bm}
\usepackage{calrsfs}             
\usepackage{tensor}

\usepackage[left=25mm,right=25mm,top=30mm,bottom=25mm]{geometry}
\usepackage{enumitem}            
\usepackage{graphicx}
\usepackage{aliascnt}            

\usepackage[
  colorlinks=true,
  linkcolor=blue,
  citecolor=blue,
  urlcolor=blue,
  pdfdisplaydoctitle=true,
  pdfstartview=FitH,
  pdfpagemode=UseNone,
  bookmarksnumbered=true,
  unicode=true,
]{hyperref}
\usepackage{hyperxmp}            

\newcommand{\newautorefthm}[3]{%
  \expandafter\def\csname#1autorefname\endcsname{#3}%
  \newaliascnt{#1}{#2}%
  \newtheorem{#1}[#1]{#3}%
  \aliascntresetthe{#1}%
}

\theoremstyle{plain}
\newtheorem{thm}{Theorem}[section]

\newautorefthm{cor}{thm}{Corollary}
\newautorefthm{lem}{thm}{Lemma}
\newautorefthm{prop}{thm}{Proposition}

\theoremstyle{definition}
\newautorefthm{defn}{thm}{Definition}
\newautorefthm{exam}{thm}{Example}

\theoremstyle{remark}
\newautorefthm{rem}{thm}{Remark}

\numberwithin{equation}{section}

\makeatletter
\def\lettername#1{\if#1\string#1#1\else\expandafter\@gobble\string#1\fi}
\makeatother
\def\myletter#1#2#3{\expandafter\newcommand\csname #1\lettername#3\endcsname{{#2 #3}}}
\def\myletterloop#1#2<#3#4>{\ifx\relax#3\else\myletter#1#2#3\myletterloop#1#2<#4>\fi}
\def\myletters#1#2#3{\myletterloop#1#2<#3\relax>}

\myletters{b}{\mathbf}{BCDEFGHIJKLMOPQRSTUVWXYZbdeghklqrstwxy}
\myletters{v}{\bm}{aefnpuvw}
\myletters{b}{\bm}{\beta\gamma\delta\epsilon\theta\kappa\mu\sigma\Sigma\tau\chi\omega}
\myletters{m}{\mathcal}{FHLMNPWU}
\myletters{b}{\pmb}{AaN\xi}
\myletter{be}{\bm}{\eta}         
\myletter{b}{\bm}{p}
\myletter{v}{\pmb}{b}

\newcommand{\RR}{\mathbb{R}}                     

\newcommand{\GL}{\mathrm{GL}}                    
\newcommand{\OO}{\mathrm{O}}                     
\newcommand{\SO}{\mathrm{SO}}                    

\newcommand{\Cinf}{\mathrm{C}^{\infty}}
\newcommand{\univ}{\mM}                          
\newcommand{\body}{\mathcal{B}}                  
\newcommand{\lab}{\mathcal{E}}                   
\newcommand{\TM}[1][\!]{T_{#1}\mM}               
\newcommand{\mat}{\mathit{rest}}

\newcommand{\Ric}{\mathbf{Ric}}                  

\newcommand{\xx}{\bx}
\newcommand{\ee}{\ve}
\newcommand{\uu}{\vu}
\newcommand{\qV}{\bq}
\newcommand{\EE}{{\pmb{E}}}
\newcommand{\FF}{{\pmb{F}}}
\newcommand{\PP}{{\pmb{P}}}
\newcommand{\UU}{{\pmb{U}}}
\newcommand{\WW}{{\pmb{W}}}
\newcommand{\XX}{{\pmb{X}}}
\newcommand{\YY}{{\pmb{Y}}}

\newcommand{\gE}{\mathfrak{g}}
\newcommand{\pp}{\mathrm{p}}
\newcommand{\bydef}{\coloneqq}
\newcommand{\pull}{\varphi^{*}}
\newcommand{\push}{\varphi_{*}}
\newcommand{\pd}[1]{\frac{\partial}{\partial {#1}}}

\DeclareMathOperator{\Diff}{Diff}                
\DeclareMathOperator{\id}{id}
\DeclareMathOperator{\dd}{d}
\DeclareMathOperator{\vol}{vol}
\DeclareMathOperator{\tr}{tr}
\DeclareMathOperator{\grad}{grad}
\DeclareMathOperator{\dive}{div}
\DeclareMathOperator{\Hess}{\mathbf{Hess}}
\DeclareMathOperator{\Wedge}{\mathsf{\Lambda}}
\DeclareMathOperator{\hodge}{\ast_{\mathnormal{g}}}

\DeclarePairedDelimiter{\norm}{\lVert}{\rVert}
\DeclarePairedDelimiter{\abs}{\lvert}{\rvert}
\DeclarePairedDelimiter{\set}{\lbrace}{\rbrace}
\DeclarePairedDelimiter{\matcomponents}{\lbrack}{\rbrack_{\mF_{\mat}}}

\hypersetup{
  pdftitle={Relativistic second gradient theory of continuous media},
  pdfauthor={Mina Chapon, Lionel Darondeau, Rodrigue Desmorat, Clément Ecker, Boris Kolev},
  pdfsubject={MSC 2020: 74B20; 83C55; 70G45},
  pdfkeywords={Relativistic elasticity, Second gradient theory of continuous media, Variational Relativity, General covariance, Objectivity.},
  pdflang=en,
}

\begin{document}

\title{Relativistic second gradient theory of continuous media}

\author[Chapon]{M.~Chapon}
\address[Mina Chapon]{Université Paris-Saclay, CentraleSupélec, ENS Paris-Saclay, CNRS, LMPS - Laboratoire de Mécanique Paris-Saclay, 91190, Gif-sur-Yvette, France}
\email{mina.chapon@ens-paris-saclay.fr}

\author[Darondeau]{L.~Darondeau}
\address[Lionel Darondeau]{Sorbonne Université, CNRS, IMJ-PRG, F-75005 Paris, France.}
\email{darondeau@imj-prg.fr}

\author[Desmorat]{R.~Desmorat}
\address[Rodrigue Desmorat]{Université Paris-Saclay, CentraleSupélec, ENS Paris-Saclay, CNRS, LMPS - Laboratoire de Mécanique Paris-Saclay, 91190, Gif-sur-Yvette, France}
\email{rodrigue.desmorat@ens-paris-saclay.fr}

\author[Ecker]{C.~Ecker}
\address[Clément Ecker]{Université Paris-Saclay, CentraleSupélec, ENS Paris-Saclay, CNRS, LMPS - Laboratoire de Mécanique Paris-Saclay, 91190, Gif-sur-Yvette, France}
\email{clement.ecker@ens-paris-saclay.fr}

\author[Kolev]{B.~Kolev}
\address[Boris Kolev]{Université Paris-Saclay, CentraleSupélec, ENS Paris-Saclay, CNRS, LMPS - Laboratoire de Mécanique Paris-Saclay, 91190, Gif-sur-Yvette, France}
\email{boris.kolev@ens-paris-saclay.fr}

\thanks{Authors’ names are in alphabetical order}

\date{October 18, 2026}%


\begin{abstract}
  Variational Relativity is a framework developed by Souriau in the 1960s to better formulate General Relativity and its classical limit, Classical Continuum Mechanics. It has been used, for instance, to formulate elasticity in General Relativity from a Lagrangian (or action). In that case, two primary variables are involved: the (Lorentzian) metric $g$ of the Universe and the matter field $\Psi$. A Lagrangian density depending on the 1-jet of these variables is then introduced, which must satisfy the Principle of General Covariance. Souriau proved in 1958 that, under these hypotheses, the Lagrangian density depends only on the pointwise value of the matter field $\Psi$ and of a secondary variable $\bK$, the conformation, which is an invariant of the diffeomorphism group and the relativistic analogue of the inverse of the right Cauchy--Green tensor. In the present work, we extend Souriau's results to a second-order gradient theory in General Relativity. Accordingly, new higher-order diffeomorphism invariants are found. Their classical limits are computed, showing that the three-dimensional second gradient theory of Continuum Mechanics can be derived from such a relativistic theory. Some of these invariants converge in the Galilean limit to objective quantities (such as the fundamental variables of finite strain second gradient theory for solids), others to non-objective quantities
  (such as those encountered in the modeling of fluids): objectivity is thus natural for the constitutive variables of solids but unnatural for those of fluids. The present work helps clarify the theoretical foundations of higher gradient theories in Continuum Mechanics.
\end{abstract}

\maketitle

\section{Introduction}

Continuum Mechanics has a long tradition of enrichment \cite{Min1964,Eri1966,ME1968,Eri1983,MM2010,AA2011,Eri1999,Pol2014,Ste2015}, by incorporating time derivatives and spatial gradients (\textit{i.e.}, jets of kinematic variables) into constitutive equations and conservation laws. Fluid Mechanics and Solid Mechanics have, however, evolved along different paths, owing to the distinct nature of the quantities and settings they address, but also to the role given to the Principle of Objectivity \cite{Nol1959,TN1965}, regarded as fundamental or not for their modeling.

The first aim of the present work is to clarify the theoretical foundations of higher gradient Continuum Mechanics:
\begin{enumerate}
  \item
        by deriving second-order jet Continuum Mechanics from General Relativity, starting from the Principle of General Covariance;
  \item
        by formulating Classical (meaning here $3$-dimensional and non-relativistic) Second Gradient Elasticity as the Galilean limit ($c\to \infty$) of the relativistic theory formulated in Minkowski spacetime~\cite{Sou1964}.
\end{enumerate}
The second aim is to provide an answer, even partial, to the question
\begin{quote}
  \emph{``Is Objectivity (Frame Indifference) a consequence of General Covariance?''}
\end{quote}

We refer to the work of Germain \cite{Ger1973,Ger1973a} for an introduction to the (non-relativistic) kinematics of higher gradient media, and to \cite{dIADG2017} for a historical account of the topic and of the legacy of Piola, Mindlin, Sedov and Toupin. We simply recall that finite strain gradient theory introduces an internal length in Solid Mechanics by considering, in the energy density, the additional kinematic variable $\bF_{\phi}^{-1}\nabla \bF_{\phi}$, which is objective and built from the spatial second derivative $\nabla \bF_{\phi}=(\partial^{2} x^{i}/\partial X^{J} \partial X^{K})$ of the deformation $\phi\colon \bX \mapsto \xx=\phi(\bX)$ (see \cite{ME1968,Ger1973,Eri1999,For2006,dSV2009}). An interesting special case is the small strain formulation of Aifantis and coworkers~\cite{Aif1992,AA1997}. Advanced constitutive theories introduce additional higher-order stress tensors~\cite{CC1909,Tou1964,Sed1968a,Sed1968,FS2003,FS2017}, possibly non-symmetric~\cite{Tou1962,Min1964} (dual to the higher-order kinematic variables; see also~\cite{dSV2009,Ber2012,BF2020} and \cite[Supplement 7.2]{Ste2015}). Note that purely spatial gradients seem to be absent from fluid constitutive laws. Lastly, the specific kinetic energy $\frac{1}{2} \norm{\partial_{t}\phi}^{2}$ is half the squared norm of the (Lagrangian) velocity $\partial_{t}\phi$, which is itself a first-order time derivative that is indeed present in the Lagrangian density of deformed media. However, higher-order time derivatives $\partial^{k}_{t}\phi$ of the deformation $\phi$ are usually not considered in the Lagrangian densities of Classical Continuum Mechanics (even when micro-inertia is accounted for~\cite{MW2007}).

Elasticity has been formulated in Special Relativity \cite{GE1966,Mau1978a,PR2013,PRA2015,NWP2022} and, since 1918, in General Relativity at the astrophysical scale \cite{Nor1916,Sou1958,Syn1959,DeW1962,Ray1963,Ben1965,CQ1972,KM1992,BS2003,EBT2006,Bro2021}. Introducing the framework of Variational Relativity~\cite{Sou1958,Sou1960,Sou1964} (adopted in \cite{CQ1972,KM1992,KM1997,BS2003}), Souriau was able to model relativistic elastic continuous media in the spirit of the gauge theories of High Energy Physics and Quantum Mechanics~\cite{HE1973,Ble1981}. The cornerstone of the theory is the representation of perfect matter by a vector-valued function $\Psi$, defined on the Universe and with values in a $3$-dimensional vector space $V \simeq \RR^3$. The $3$-dimensional body $\body$ of Classical Continuum Mechanics~\cite{TN1965} is then recovered as a compact submanifold with boundary of $V$ endowed with a mass measure $\mu$ that labels the material particles. Its preimage by $\Psi$, the Matter World tube, corresponds to the whole evolution of the considered matter in the Universe. Note that this description of continuous media is the reverse of the one generally used in Classical Continuum Mechanics, where configurations are described by embeddings $\mathrm{p} : \body \to \lab$, where $\lab$ is the three-dimensional space of Classical Mechanics\footnote{See \cite{GG1999,dIMS2009} for counterexamples.}. However, this description will be recovered once an observer (or time function) is introduced, which leads to a foliation of the Universe by spatial hypersurfaces (corresponding to the usual three-dimensional space $\lab$ \cite{Gou2012,KD2023}). Souriau's formulation of elasticity does not require a time function \textit{a priori}. He can therefore express the consequences of demanding General Covariance of the matter Lagrangian. Assuming that the matter Lagrangian density depends only on the zero-order jet of the metric $g$ and on the first-order jet $(\Psi,T\Psi)$ of the perfect matter field $\Psi$, Souriau's 1958 theorem~\cite{Sou1958} states that General Covariance implies that this density depends only on the pointwise value of $\Psi$ and of the secondary variable
\[
  \bK
  \bydef
  T\Psi\left.g^{-1}\right.T\Psi^{\star},
\]
which he named the \textsl{conformation}, where $(\cdot)^{\star}$ denotes the dual transpose.

\begin{thm}[Souriau, 1958]\label{thm:first_gradient}
  Suppose that the Lagrangian defined on the $4$-dimensional Universe,
  \[
    \mL^{\text{matter}}[g, \Psi]
    =
    \int
    L_{0}\left(\tensor{g}{_{\mu\nu}}, \tensor{\Psi}{^{I}}, \partial_{\lambda}\Psi^{I}\right)
    \vol_{g},
  \]
  is generally covariant, where $g$ is a Lorentzian metric, $\vol_{g}$ the corresponding volume form and $\Psi$ a matter field. Then, its Lagrangian density $L_{0}$ can be recast as $L_{0}=L(\Psi, \bK)$, for some function~$L$.
\end{thm}

In the present work, we extend Souriau's theorem to a second-order gradient theory in General Relativity. The starting point is the introduction of a second-order jet Lagrangian density, depending only on the two primary variables $g$ (the Universe metric) and $\Psi$ (the perfect matter field), and assumed to be generally covariant~\cite{Sou1958,Sou1964,KM1997,BS2003}. Accordingly, we define secondary variables (which are invariants of the diffeomorphism group) that play the same role as Souriau's conformation, but for a relativistic second gradient theory. Among them, we exhibit three novel gravitation/matter field coupling tensors (see \autoref{sec:second_gradient}).

The central idea of this work is to find a systematic way to extract relativistic invariants from generally covariant tensors, in the spirit of Souriau's definition of $\bK$, which is a vector-valued relativistic invariant. The main contribution of this work is the observation that the rest frame, properly formulated, is generally covariant of the metric $g$ and the matter field $\Psi$. This allows us first to give a simple proof of Souriau's theorem and then to prove our result with little effort. Our methods can be extended without limit to higher-order theories or to a richer physical structure of the Universe, but they require a (perfect) matter field (this limits the scope of our technique to interior solutions of General Relativity, \textit{i.e.}, solutions inside matter).

Finally, we show that the classical $3$-dimensional second gradient theory can be derived from such a relativistic theory as its Galilean approximation. Based on our General Covariance result, we address the issue of the status of the Principle of Objectivity in Classical (non-relativistic) Continuum Mechanics. Specifically, we show that the spatial gradient-type General Covariant secondary variables (modeling the deformation of solids) converge to objective quantities in the Galilean limit, while time-derivative-type General Covariant secondary variables may converge to non-objective quantities in the Galilean limit.

The outline of the paper is as follows. In~\autoref{sec:GR_framework}, we recall some basic concepts in Variational Relativity: General Covariance, Lagrangian action functionals, matter field and matter rest frame. In~\autoref{sec:lagrangians}, we state and prove our main result, \autoref{thm:second_gradient}. We also provide an application to gradient fluids and a new, simple proof of Souriau's theorem. In~\autoref{sec:observer-frames}, we introduce the concepts of time function and observer frame, and the associated foliation by spacelike hypersurfaces (the so-called space-time structure). We rewrite our relativistic invariants using this space-time decomposition and give useful specializations for Minkowski spacetime. Finally, in~\autoref{sec:classical-limit}, we compute their classical limits and discuss the question of objectivity (frame-indifference) of these limits. In addition, to be as self-contained as possible, we have added two appendices summarizing the main mathematical concepts and formulas used in this paper.

\section{Matter fields in Variational Relativity}
\label{sec:GR_framework}

The Universe is assumed to be a $4$-dimensional orientable manifold $\univ$, endowed with a hyperbolic metric $g$, of signature $(-+++)$.

\subsection{General Covariance}

The main postulate of General Relativity is that \emph{Physical laws must be independent of the choice of coordinates}~\cite{Ein1988}. This principle is known as \textsl{General Covariance}, or invariance under reparameterization. To lift the ambiguity created by the imprecise wording of this principle~\cite{Nor1993,WS2009}, let us make clear what is meant here.

The Universe $\univ$ is locally diffeomorphic to $\RR^{4}$, \textit{via} either \textsl{coordinate charts}
\[
  \bx\colon\univ\to\RR^{4}
\]
or via \textsl{parameterizations}
\[
  \bxi\colon\RR^{4}\to\univ,
\]
also called \emph{plots} by Souriau and Iglesias-Zemmour in a more general setting, see~\cite{Sou1980,Igl2013}. Of course, parameterizations are the inverse of coordinate charts and \textit{vice versa}. A \textsl{reparameterization at the source} (or \textsl{change of coordinates}) is a local diffeomorphism $\varPhi$ from $\RR^{4}$ to $\RR^{4}$, whereas a \textsl{reparameterization at the goal} (or plainly a \textsl{reparameterization}) is a local diffeomorphism $\varphi$ from $\univ$ to $\univ$. Both changes of coordinates and reparameterizations act, by composition, on either coordinate charts or parameterizations. However, whereas adopting the point of view of coordinate charts or of parameterizations is a strictly equivalent choice, this is no longer the case for changes of coordinates and reparameterizations (see the figures and the explanations below).
\begin{figure}[htbp]
  \includegraphics[width=.49\textwidth]{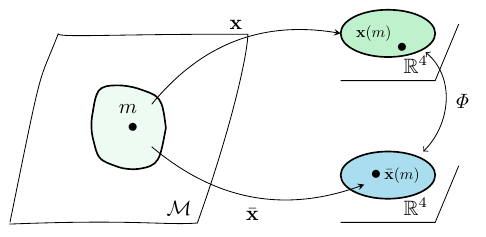}
  \includegraphics[width=.49\textwidth]{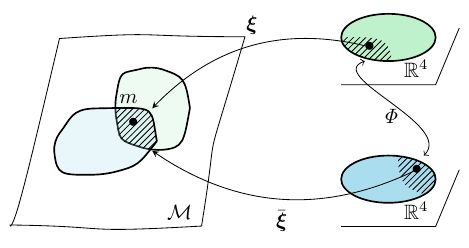}
  \caption{A change of coordinates (left), and a reparameterization at the source  (right).}
\end{figure}
\begin{figure}[htbp]
  \includegraphics{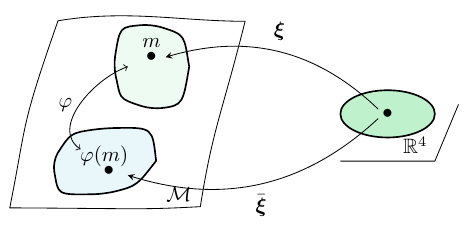}
  \caption{A reparameterization (at the goal).}
\end{figure}

Firstly, there is no \emph{canonical} equivalence between changes of coordinates and reparameterizations: one needs to fix an arbitrary (local) diffeomorphism $\RR^{4}\to\univ$. Secondly and more importantly, even after identification, the actions of the two diffeomorphism groups are different. Changes of coordinates are involved in the process of gluing \emph{local objects} into \emph{global objects} such as tensor fields, whereas changes of parameterization are used to compare \emph{global objects} at different positions by sending them into the same \emph{local chart}.

One often aims to formulate theories using intrinsic formulas involving tensor fields (which is checked using changes of coordinates). General covariant tensor fields themselves depend on other tensor fields in a generally covariant manner (which is checked using changes of parameterization).

Now that the principle is clearly stated, let us be a little more concrete. Variational Relativity, as formulated by Souriau~\cite{Sou1958,Sou1964}, involves \textsl{primary} variables (for instance, the metric $g$), and \textsl{secondary} variables, which depend on the primary variables in a covariant manner, for instance the Riemannian covariant derivative $\nabla^{g}$, the Riemann curvature tensor $\bR^{g}$, the Ricci tensor $\Ric^{g}$ or the scalar curvature $R^{g}$ (see \autoref{sec:nabla-and-tensors}). Acting in a covariant manner means that an action of the diffeomorphism group (denoted $\star$, see \autoref{subsec:diff-action-tensors}) on the primary variable $\bP$ and the secondary variable $\bS$ has been defined in such a way that
\[
  \bS(\varphi \star \bP) = \varphi \star \bS(\bP),
\]
for all local diffeomorphisms $\varphi$ of the Universe $\univ$. We then say that $\bS$ is a \textsl{covariant} of $\bP$. An important special case is when $\bS$ is a scalar field or a vector-valued function and the action of the diffeomorphism group on $\bS$ is given by
\[
  (\varphi \star \bS)(m) = \bS(\varphi(m)).
\]
In that case, we get
\begin{equation}\label{eq:relativistic-invariant}
  \bS(\varphi \star \bP)(m) = \bS(\bP)(\varphi(m)).
\end{equation}
We shall then say that $\bS$ is a \emph{relativistic invariant}. Note that the scalar curvature $R^{g}$ is a relativistic invariant, but the Riemann curvature tensor $\bR^{g}$ and the Ricci tensor $\Ric^{g}$ are not.

\subsection{Variational Relativity}

We now apply the principle of General Covariance in the context of the Lagrangian field theory of General Relativity~\cite{Sou1958,Sou1964,HE1973}. In Lagrangian field theory, the primary variables are field variables defined on the Universe $\univ$. The Lagrangian functional (or action)
\[
  \mL[\bP] \bydef\int L(j_{m}^{n}\bP)\,\vol_{g}
\]
is defined using a Lagrangian density $L$ which depends on the pointwise values (at $m \in \univ$) of the primary variables $\bP$ and their derivatives up to a given order $n$ (what is traditionally called their $n$th order jet and denoted by $j_{m}^{n} \bP$). Usually, when the Lagrangian satisfies the principle of General Covariance, its Lagrangian density $L$ does not depend on the whole jet of $\bP$ but rather on the $n$th order jet of $\bP$ through secondary variables $\bS$ which are covariants of the diffeomorphism group. For instance, in the initial variational formulation of General Relativity (in the \textit{vacuum}), there is a unique primary variable, the metric $g$, and the Lagrangian density is given by
\[
  L(j_{m}^{2}g) = L(g_{\mu\nu}(m),\partial_{\kappa}g_{\mu\nu}(m), \partial^{2}_{\lambda\kappa}g_{\mu\nu}(m)) = \frac{1}{2\kappa}\left( R^{g}(m) -\Lambda\right),
\]
where $\kappa$ is the Einstein constant and $\Lambda$ is the cosmological constant. Here, $R^{g}$ is the scalar curvature of $g$ which is a secondary variable and a covariant of $g$, meaning that $R^{\pull g}=\pull R^{g}$ (see \autoref{subsec:diff-action-tensors}).

In this work, we will consider Lagrangian functionals $\mL$ that depend on two primary variables: the Lorentzian metric $g$ of the Universe $\univ$ and a vector-valued function $\Psi\colon\univ\to\RR^3$, called the \emph{matter field}~\cite{Sou1958,Sou1964,KD2023}. To avoid unnecessary analytical difficulties, and since in practice we do not require Lagrangian densities $L$ to be integrable over the whole manifold $\univ$ (which is usually not compact), Lagrangian densities are integrated only over relatively compact domains $\mU \subset \univ$ (furthermore contained in a local chart).
Therefore, we shall write
\[
  \mL_{\mU}[g,\Psi] = \int_{\mU} L(j_{m}^{\ell}g,j_{m}^{n}\Psi) \, \vol_{g},
\]
to emphasize the dependence on $\mU$.

Let us describe the principle of General Covariance for the Lagrangian $\mL$ in more precise terms and formulate its consequences. Let $\varphi \colon \mU \to \bar{\mU}$ be a diffeomorphism between two open sets $\mU$ and $\bar{\mU}$ of the Universe $\univ$ and let $T\varphi \colon T\mU \to T\bar{\mU}$ be its tangent map. Then, the Lagrangian $\mL$ is invariant by $\varphi$ if
\[
  \mL_{\mU}[\pull g,\pull\Psi] = \mL_{\bar{\mU}}[g,\Psi],
\]
for every Lorentzian metric $g$ on $\univ$, every vector-valued function $\Psi\colon\univ \to \RR^3$, and every local diffeomorphism $\varphi \colon \mU \to \bar{\mU}$. Here, $\pull$ means the action by pullback of the local diffeomorphism $\varphi$ (see \autoref{subsec:diff-action-tensors}).
One has
\[
  \pull g \bydef (T\varphi)^{\star}\,(g \circ \varphi) (T\varphi), \qquad  \pull\Psi \bydef \Psi \circ \varphi,
\]
where $(\cdot)^{\star}$ is the dual transpose. Using a change of variables, we thus observe that the Lagrangian is invariant by $\varphi$ if
\[
  \int_{\mU} L(j_{m}^{\ell}(\pull g),j_{m}^{n}(\pull\Psi)) \, \vol_{\pull g}
  =
  \int_{\bar{\mU}} L(j_{\bar m}^{\ell} g,j_{\bar m}^{n} \Psi) \, \vol_{g}
  =
  \int_{\mU} L(j_{\varphi(m)}^{\ell} g,j_{\varphi(m)}^{n}\Psi) \, \pull\vol_{g}.
\]

Since the open set $\mU$ is arbitrary and $\vol_{\pull g} = \pull\vol_{g}$, one finally obtains the following condition on the Lagrangian density, for arbitrary local diffeomorphisms $\varphi$:
\begin{equation}\label{eq:lagrangian_density}
  L(j_{m}^{\ell}(\pull g),j_{m}^{n}(\pull\Psi)) = L(j_{\varphi(m)}^{\ell} g,j_{\varphi(m)}^{n}\Psi).
\end{equation}

Since general covariance of the Lagrangian functional is equivalent to condition~\eqref{eq:lagrangian_density}, we obtain the following elementary result, on which our general reasoning in \autoref{sec:lagrangians} will rely.

\begin{lem}
  A Lagrangian functional is generally covariant if and only if its density is a relativistic invariant. As a consequence, a Lagrangian whose density depends only on relativistic invariants (definition~\eqref{eq:relativistic-invariant}) is automatically generally covariant.
\end{lem}

\subsection{Matter field}
\label{subsec:matter_field}

In Variational Relativity, following Souriau~\cite{Sou1958,Sou1964}, matter is described by the set of all \textsl{lifelines} (or \textsl{World lines}) of the material points constituting the continuous medium under study. More precisely, one considers a smooth vector-valued function
\[
  \Psi\colon\univ\to V \simeq \RR^{3}.
\]
The labels of these particles constitute a subset $\body\subseteq V$, in general assumed to be a $3$-dimensional compact orientable manifold with boundary and called the \textsl{body}. It is further assumed that $\Psi$ is a submersion on
\[
  \mW\bydef\Psi^{-1}(\body),
\]
meaning that the linear tangent map
\[
  T\Psi\colon T\mW\to TV
\]
is of rank $3$ at each point of $\mW$. Thus, $\mW$ is fibered by the particle lifelines $\Psi^{-1}(\mathbf{X})$, for $\mathbf{X}\in\body$, and is called for this reason the body's \textsl{World tube}. It is moreover assumed (\textsl{perfect matter} assumption) that each $1$-dimensional subspace $\ker T_{m}\Psi$ of $T_{m}\univ$ is \textsl{timelike}, meaning that the restriction of the metric $g$ to these subspaces is negative definite.

The body $\body$ is equipped, as in Classical Continuum Mechanics, with a mass measure $\mu$ (in our case, a volume form). Using the canonical coordinate system $(X^{I})$ and the canonical metric $\qV=(\delta_{IJ})$ of $V\simeq\RR^3$, $\mu$ can be recast as
\[
  \mu
  =
  \rho_{\circ} \, \vol_{\qV},
  \qquad
  \vol_{\qV}= \left.\dd X^1\wedge\dd X^2\wedge\dd X^3\right.,
\]
where $\rho_{\circ}$ is a positive function supported on the body and interpreted as a mass density on $\body$. The integral
\[
  m=\int_{\body} \mu
\]
is itself interpreted as the total mass of the continuous medium under study.

The pullback $\Psi^{*} \mu$ by $\Psi$ of this mass measure $\mu$ is a $3$-form defined on the $4$-dimensional World tube $\mW=\Psi^{-1}(\body)$. Since $\Psi$ is assumed to be a submersion on $\mW$, there exists a \emph{nowhere vanishing} vector field $\PP$ on $\mW$, such that
(see~\cite{Sou1958,Sou1964,KD2023}, \textit{cf.} also \eqref{eq:W})
\begin{equation}\label{eq:defP}
  i_{\PP} \vol_{g}
  =
  \Psi^{*}\mu,
\end{equation}
where $i_{\PP}$ is the interior product (or left contraction) with $\PP$. This vector field $\PP$ is called the \textsl{current of matter}.
It spans $\ker(T\Psi)$ and is thus timelike.
The \textsl{rest mass density} on the World tube $\mW$ is the scalar
\[
  \rho_r \bydef \sqrt{-g(\PP, \PP)} \, ,
\]
and we define the \emph{four-velocity} $\UU$ (this will be clarified with the introduction of a spacetime in \autoref{sec:classical-limit}) as the unit timelike vector field on $\mW$ such that
\begin{equation}\label{eq:U}
  \PP=\rho_{r} \UU,
  \qquad
  g(\UU,\UU) = -1.
\end{equation}

\begin{rem}[Mass conservation, first expression]
  Since $\dd \mu=0$ and $\left(\dive^{g} \PP\right) \vol_{g}=\dd ( i_{\PP} \vol_{g})=\dd\Psi^{*} \mu=\Psi^{*}\dd \mu$, a first (kinetic) relativistic expression of mass conservation is
  \[
    \dive^{g}\PP=\dive^{g}(\rho_{r} \UU)=0.
  \]
\end{rem}

The \emph{conformation tensor} $\bK$ is a vector-valued function defined on $\univ$ with values in $S^{2}\RR^3$ and given by
\begin{equation}\label{eq:K}
  \bK \bydef T\Psi \,g^{-1} T\Psi^{\star},
\end{equation}
where $T\Psi^{\star}$ is the dual map (or transpose) of $T\Psi$. Note that $\bK(m)$ is a symmetric contravariant tensor on $\RR^3$ which is positive definite for each $m \in \mW$ for perfect matter (see \cite[Theorem p. 9]{Lic1955} and \cite{Sou1964,KD2023}).

We denote by $ \Psi^{I}$ the components of the vector-valued matter field $\Psi\colon\univ\to \RR^3$ in the canonical coordinate system $(X^{I})$ of $\RR^3$.
In this coordinate system, the components of the conformation are
\[
  K^{IJ}
  =
  g^{-1}(\dd\Psi^{I},\dd\Psi^J).
\]

Introducing the four-vector field $\WW$, defined implicitly on $\mM$ by
\begin{equation}\label{eq:definiW}
  i_{\WW}\vol_{g} =\Psi^{*} \vol_{\bq},
\end{equation}
so that the current of matter can be written as $\PP = (\rho_{\circ} \circ \Psi) \WW$ on $\mW$, and applying \eqref{eq:norm_W}, we recover that
\[
  g(\PP,\PP)
  =
  -(\rho_{\circ} \circ \Psi)^{2}
  \det\big(g^{-1}(\dd\Psi^I,\dd\Psi^J)\big).
\]
and the following expression of the rest mass density due to Souriau \cite{Sou1964}
\begin{equation}\label{eq:rhor}
  \rho_{r} = (\rho_{\circ} \circ \Psi) \sqrt{\det(\bK \qV)},
\end{equation}
where $\det(\bK \qV)=\det(K^{IJ})$ is the determinant of the matrix $(K^{IJ})$.

\begin{rem}[Mass conservation, second expression]
  The expression \eqref{eq:rhor} of the rest mass density is a second (integrated) relativistic expression of mass conservation. It is a generalization of the Classical Continuum Mechanics expression $\rho \circ \phi= \rho_{0} / \sqrt{\det (C^{IJ})}$, where $\rho_{0}$ and $\rho$ are, respectively, the mass densities on the 3D reference configuration $\Omega_{0}$ and the deformed configuration $\Omega=\phi(\Omega_{0})$, and where $\bC$ is the right Cauchy--Green tensor (defined on~$\Omega_{0}$).
\end{rem}

\subsection{Matter rest frame}
\label{sec:matter_frame}

The vector field $\WW$ defined implicitly on the Universe $\mM$ by
\[
  i_{\WW}\vol_{g} = \dd \Psi^{1} \wedge \dd \Psi^{2} \wedge \dd \Psi^{3}
\]
is non-zero at each point $m \in \mM$ at which $\Psi$ is a submersion. It is therefore non-vanishing on the world-tube $\mW$.

In Appendix~\ref{subsec:Hodge}, we prove that the covector field $\WW^{\flat}=g\WW$ is given by
\begin{equation}\label{eq:defW}
  \WW^{\flat}\bydef\hodge( \dd\Psi^{1} \wedge \dd\Psi^{2} \wedge \dd\Psi^{3}),
\end{equation}
where $\hodge$ denotes the Hodge star operator for the metric $g$, and therefore that $\WW$ is orthogonal to the subspace spanned by the vector fields
\[
  \grad^{g}\Psi^{I} = (\dd\Psi^{I})^{\sharp} = g^{-1}\dd\Psi^{I}, \qquad I=1,2,3.
\]

Since $ \WW $ is timelike on the world-tube $\mW$ (perfect matter hypothesis), we furthermore conclude that $ \langle \WW \rangle^{\bot}$ is spacelike (see \cite[Theorem p. 9]{Lic1955}) and thus that the family
\[
  (\WW, \grad^{g}\Psi^{1}, \grad^{g}\Psi^{2}, \grad^{g}\Psi^{3})
\]
is a frame of $T_{m}\mM$ at each point $m \in \mW$.

In Appendix~\ref{subsec:Hodge}, we show moreover that $\WW^{\flat}$ (and $\WW$) depends on the pointwise value of the Lorentzian metric $g$ and on the first-order jet of the matter field $\Psi$. It is furthermore generally covariant, meaning that
\[
  \WW^\flat(\varphi^{*}g,\varphi^{*}\Psi, \varphi^{*}T\Psi) = \varphi^{*} \WW^\flat(g,\Psi, T\Psi),
\]
for every local diffeomorphism $\varphi$ of $\mM$.

\begin{rem}
  Whereas the covector fields $\dd\Psi^{I}$ are naturally covariant under the group of general diffeomorphisms, the supplementary covector field $\WW^{\flat}=\WW^{\flat}(g,\Psi, T\Psi)$ changes sign when the orientation is reversed (see Appendix~\ref{subsec:volume-pseudo-form} for more details).
\end{rem}

This covariance property is inherited by the unit timelike vector field
\[
  \UU = \frac{\WW}{\sqrt{-g(\WW,\WW)}}
\]
and by its covariant version $\UU^{\flat}$, since $g(\WW, \WW)$ is a relativistic invariant. Therefore, we shall introduce the coframe
\[
  \mF_{\text{rest}}^{\star}
  \bydef
  \left(-\UU^\flat, \EE^{1}, \EE^{2}, \EE^{3}\right),
  \qquad
  \EE^{I}:= \dd\Psi^{I},
\]
defined at each point $m\in\mW$, and its dual frame $\mF_{\mathrm{rest}}$
\[
  \mF_\mat
  =
  \left(\UU,\EE_1,\EE_2,\EE_3\right).
\]
One has thus for $I,J=1,2,3$
\begin{equation}\label{eq:E_I}
  \UU^\flat(\EE_I)=0
  \quad\text{and}\quad
  \dd\Psi^I(\EE_J)=\tensor{\delta}{^I_J}.
\end{equation}
This frame $\mF_{\mathrm{rest}}$ shall be referred to as the \emph{matter rest frame}, or simply as the \emph{rest frame}.

\begin{rem}
  Since $T_{m}\Psi \colon T_{m}\mU \to V$ is surjective, it has \emph{right inverses}, which means linear mappings $\bL_{m} \colon V \to T_{m}\mU$, such that $T_{m}\Psi \cdot \bL_{m} = \id_{V}$. Each of them depends on the choice of a complement $S$ to $\ker T_{m}\Psi$ and will be denoted by $(T_{m}\Psi)_{S}^{-1}$. Among them, there is one which is ``nicer'' than the others, the one corresponding precisely to $(\ker T_{m}\Psi)^{\bot}$ (which is a complement of $\ker T_{m}\Psi$ because we have assumed that $\ker T_{m}\Psi$ is timelike on $\mU$). This one will be called the \emph{orthogonal right inverse of $T_{m}\Psi$} and denoted by $(T_{m}\Psi)_{\UU^{\perp}}^{-1}$. Using this definition, one can observe that the vector fields $\EE_{J}$ can be written as
  \[
    \EE_{J}= (T_{m}\Psi)_{\UU^{\perp}}^{-1}\cdot \frac{\partial}{\partial X^{J}},
  \]
  introducing the canonical basis $(\partial/\partial X^{I})$ on $V\simeq \RR^{3}$. Indeed, we have
  \[
    \UU^{\flat}(\EE_{J})=0, \quad \text{and} \quad \EE^{I}(\EE_{J}) = \dd X^{I} \cdot T_{m}\Psi \cdot (T_{m}\Psi)_{\UU^{\perp}}^{-1} \cdot \frac{\partial}{\partial X^{J}} = \delta^{I}_{J}.
  \]
\end{rem}

We shall summarize all these considerations in the following theorem.

\begin{thm}\label{thm:rest-frame}
  Assume that $\Psi\colon\mM \to V$ is a submersion on $\mW$ and that $\ker T_{m}\Psi$ is timelike at each point $m\in\mW$. Then, the frame
  \[
    \mF_{\mathrm{rest}} = \big(\UU, \EE_{1}, \EE_{2}, \EE_{3}\big),
  \]
  defined on $\mW$, and its dual frame
  \[
    \mF_{\mathrm{rest}}^{\star} = \big(-\UU^{\flat}, \EE^{1}, \EE^{2}, \EE^{3}\big),
  \]
  where $\EE^{I}\bydef\dd \Psi^{I}$, depend only on the pointwise value of the Lorentzian metric $g$ and of the first-order jet of the matter field $\Psi$.
  Moreover, given an orientation-preserving local diffeomorphism $\varphi$ of $\mW$, one has
  \begin{align*}
     & \varphi^{*}\mF_{\mathrm{rest}}(g,\Psi, T\Psi) = \mF_{\mathrm{rest}}(\varphi^{*}g,\varphi^{*}\Psi, \varphi^{*}T\Psi),
    \\
     & \varphi^{*} \mF_{\mathrm{rest}}^{\star} (g,\Psi, T\Psi)= \mF_{\mathrm{rest}}^{\star} (\varphi^{*}g,\varphi^{*}\Psi, \varphi^{*}T\Psi).
  \end{align*}
  In other words, the frame $\mF_{\mathrm{rest}}$ and its coframe $\mF_{\mathrm{rest}}^{\star}$ are generally covariant.
\end{thm}

\section{General covariant Lagrangians involving matter fields}
\label{sec:lagrangians}

The introduction of the rest frame is essential to build relativistic invariants, similar to the components of the conformation $\bK$, and thereby to formulate a generally covariant theory of gradient media in General Relativity.

\subsection{First Gradient Relativistic Elasticity}

Relativistic elasticity has been formulated by Souriau in~\cite{Sou1958,Sou1964} starting with a matter Lagrangian density
\[
  L_{0}
  =
  L_{0}\left(\tensor{g}{_{\mu\nu}}(m), \tensor{\Psi}{^{I}}(m), \partial_{\lambda}\tensor{\Psi}{^{I}}(m)\right),
\]
defined at each point $m\in\mM$. It is a function of the pointwise values $g(m)$ of the metric, of the matter field $\Psi(m)$ and of its first derivative $T_{m} \Psi$. Souriau has shown that General Covariance of the matter Lagrangian implies that its density can be recast as a function of $\Psi(m)$ and of the conformation $\bK(m)$ (\autoref{thm:first_gradient}). To illustrate the ideas of our new method---which extends straightforwardly to higher gradients---we first give an alternative, simpler proof of Souriau's 1958 \autoref{thm:first_gradient}, valid only on the World tube $\mW$. This is common in General Relativity when one looks for \emph{interior solutions} of Einstein's equations (see \cite{MTW1973,MG2009}), that is, solutions defined \emph{inside matter}.

\begin{proof}[Proof of \autoref{thm:first_gradient}]
  Starting with the variables $(g,\Psi,T\Psi)$, we first construct the coframe $\mF_{\mat}^{\star}(g,\Psi,T\Psi)$ as in \autoref{subsec:matter_field}.
  Since it amounts to adding the covector field $\UU^\flat(g,\Psi,T\Psi)$ to $T\Psi$, one can safely replace the variables $(g,\Psi,T\Psi)$ by $(g,\Psi,\mF_{\mat}^{\star})$.

  One can now replace the metric tensor $g$ by the matrix of the components of the cometric $g^{-1}$ in the coframe $\mF_{\mat}^{\star}$. Taking \eqref{eq:K} into account, this amounts to replacing the pointwise value of the metric tensor $g$ by that of the conformation tensor $\bK$, whose components $(K^{IJ})$ in the rest frame are relativistic invariants.

  Hence, without having yet used the general covariance of the Lagrangian, we have been able to recast the function $L_{0}(g,\Psi,T\Psi)$ as a new function $L(\Psi,\bK,\mF_{\mat}^{\star})$, depending on the relativistic invariants $\Psi$ and $\bK$, and on the generally covariant coframe $\mF_{\mat}^{\star}$. The general covariance of the Lagrangian density hence becomes
  \[
    L(\Psi,\bK,\mF_{\mat}^{\star})
    =
    L(\Psi,\bK,\pull\mF_{\mat}^{\star}),
  \]
  for any orientation-preserving local diffeomorphism $\varphi$ of $\univ$.

  Since the action of the (orientation-preserving) diffeomorphism group on the pointwise value of the rest frame $\mF_{\mat}^{\star}$ involves only the standard action of the orientation-preserving linear group $\varphi\in\GL^{+}(4,\RR)$ which acts transitively on the oriented bases of $\RR^{4}$, we conclude that $L$ does not depend on $\mF_{\mat}^{\star}$. This completes the proof.
\end{proof}

\begin{rem}
  Note \textit{a posteriori} that the result is still true if one considers covariance under the group of general diffeomorphisms instead of only orientation-preserving diffeomorphisms, because the vector field $\UU$ is not involved in the formation of the relativistic invariants $\Psi$ and $K^{IJ}$. This will no longer be the case for higher-order gradients.
\end{rem}

\begin{rem}\label{rem:PushCometric}
  Choosing the components of the cometric rather than those of the metric is natural. Indeed, since $T\Psi$ is not invertible, one can pushforward a contravariant tensor, but not a covariant tensor (see \autoref{subsec:diff-action-tensors}).
\end{rem}

\subsection{Second Gradient Relativistic Elasticity}
\label{sec:second_gradient}

We now seek a generally covariant formulation of Relativistic Gradient Elasticity, starting with a matter Lagrangian density of general form
\[
  L_{0}(
  \tensor{g}{_{\mu\nu}}, \partial_{\nu}\tensor{g}{_{\lambda\mu}},\partial_{\mu}\partial_{\nu}\tensor{g}{_{\kappa\lambda}},
  \tensor{\Psi}{^{I}}, \partial_{\lambda}\tensor{\Psi}{^{I}},\partial_{\lambda}\partial_{\mu}\tensor{\Psi}{^{I}}
  ),
\]
defined at each point $m\in\mM$ of the Universe. This Lagrangian density is a function of the second-order jet of the metric $g$ and of the second-order jet of the matter field $\Psi$. We shall assume, furthermore, that $ \Psi $ is a submersion and $\ker(T\Psi)$ is timelike.

\begin{thm}\label{thm:second_gradient}
  Consider the Lagrangian
  \begin{equation}\label{eq:matterL}
    \mL^{\text{matter}}[g, \Psi]
    =
    \int
    L_{0}(
    \tensor{g}{_{\mu\nu}}, \partial_{\nu}\tensor{g}{_{\lambda\mu}},\partial_{\mu}\partial_{\nu}\tensor{g}{_{\kappa\lambda}},
    \tensor{\Psi}{^{I}}, \partial_{\mu}\tensor{\Psi}{^{I}},\partial_{\mu}\partial_{\nu}\tensor{\Psi}{^{I}})\vol_{g}
  \end{equation}
  depending on the field variables $g$, a Lorentzian metric, and $\Psi \colon \mM \to V$, a submersion such that $\ker(T\Psi)$ is timelike. Then, if $\mL^{\text{matter}}$ is generally covariant under the group of orientation-preserving diffeomorphisms, its Lagrangian density $L_{0}$ can be recast, at each point of the World tube $\mW$, as
  \begin{equation}\label{eq:matterL0}
    L_{0} = L\left(\matcomponents{g^{-1}},\matcomponents{\bR^{g}},\Psi,\matcomponents{\Hess^{g} \Psi}\right),
  \end{equation}
  for some function $L$, where $\bR^{g}$ is the curvature tensor, $\Hess^{g} \Psi$ is the Hessian tensor of the vector-valued function $\Psi$, and where $\matcomponents{\bt}$ denotes the set of components of a tensor field $\bt$ in $\mF_{\mat}$.
\end{thm}

\begin{rem}
  Theorem~\ref{thm:second_gradient} shows that a generally covariant Lagrangian depending on the 2-jets of a metric $g$ and a matter field $\Psi$ reduces to a function of a few diffeomorphism invariants (components of $g^{-1}$, curvature and Hessian of $\Psi$ in the matter rest frame). This is an instance of the reduction principle of Kolář–Michor–Slovák~\cite{KSM1993}: a natural operator of order $r$ corresponds to an equivariant map on jets at a point. The proof follows that pattern. Diffeomorphisms with trivial linear part act trivially on tensors, so they can be used to normalize the Christoffel symbols and replace the metric jets by curvature. What remains is an invariant-theory problem for $\GL^{+}(4,\RR)$, which the natural rest frame solves directly.
\end{rem}

\begin{proof}
  The proof is very similar to our new proof of \autoref{thm:first_gradient}.

  \begin{enumerate}[wide]
    \item
          Starting with the variables
          \[
            \left(
            \tensor{g}{_{\mu\nu}}, \partial_{\nu}\tensor{g}{_{\lambda\mu}},\partial_{\mu}\partial_{\nu}\tensor{g}{_{\kappa\lambda}},
            \tensor{\Psi}{^{I}}, \partial_{\mu}\tensor{\Psi}{^{I}},\partial_{\mu}\partial_{\nu}\tensor{\Psi}{^{I}}
            \right),
          \]
          a first change of variables consists in substituting the first and second derivatives of $g_{\mu\nu}$ by the Christoffel symbols of the corresponding Levi-Civita connection and their derivatives (see \autoref{sec:nabla-and-tensors}). The variables are now
          \[
            \left(
            \tensor{g}{_{\mu\nu}},
            \tensor{\Gamma}{^{\lambda}_{\mu\nu}}, \partial_{\nu}\tensor{\Gamma}{^{\kappa}_{\lambda\mu}},
            \tensor{\Psi}{^{I}}, \partial_{\mu}\tensor{\Psi}{^{I}}, \partial_{\mu}\partial_{\nu}\tensor{\Psi}{^{I}}
            \right).
          \]

    \item
          The Levi-Civita connection allows one to define the metric covariant derivative $\nabla^{g}$ of a tensor field. In particular, one can define the Hessian of the scalar function $\Psi^{I}$ (see \autoref{sec:nabla-and-tensors})
          \[
            \Hess^{g} \Psi^{I}
            \bydef
            \nabla^{g} \dd \Psi^{I},
          \]
          a second-order covariant symmetric tensor field, with components
          \[
            \tensor{(\Hess^{g}\Psi)}{^{I}_{\mu\nu}} = \partial_{\mu}\partial_{\nu} \Psi^{I} - \tensor{\Gamma}{^{\rho}_{\mu\nu}}\partial_{\rho}\Psi^{I}.
          \]
          This allows us to replace (using a change of variables) the second-order derivatives of $\Psi^{I}$ by the components of its covariant Hessian. The variables are now
          \[
            \left(
            \tensor{g}{_{\mu\nu}},
            \tensor{\Gamma}{^{\lambda}_{\mu\nu}}, \partial_{\nu}\tensor{\Gamma}{^{\kappa}_{\lambda\mu}},
            \tensor{\Psi}{^{I}}, \partial_{\mu}\tensor{\Psi}{^{I}}, \tensor{(\Hess^{g} \Psi)}{^{I}_{\mu\nu}}
            \right).
          \]

    \item
          So far, we have not used the general covariance hypothesis. At this stage, we have replaced most of the jet variables by the local components of generally covariant tensors, with the exception of $\tensor{\Gamma}{^{\lambda}_{\mu\nu}}$ and $\partial_{\nu}\tensor{\Gamma}{^\kappa_{\lambda\mu}}$.
          We will now use the general covariance of the Lagrangian and the normalization result \autoref{thm:normalization} established in \autoref{subsec:diff-action-jets} in order to replace these remaining jet variables by components of generally covariant tensors.

          Consider the subgroup of $\Diff^{+}(\univ)$ made of diffeomorphisms with trivial linear part, \textit{i.e.}
          \[
            \set*{\varphi\in\Diff(\univ)\colon\varphi(m)=m\text{ and }T_m\varphi=\mathrm{id}}.
          \]
          If $\bt(g, \Psi)$ is a generally covariant tensor field, then for such a diffeomorphism $\varphi$, one has the pointwise invariance property
          \[
            \bt(\pull g,\pull\Psi)
            =
            \pull\bt(g,\Psi)
            =
            \bt(g,\Psi),
          \]
          because, pointwise, diffeomorphisms act on tensors via their linear parts.
          By \autoref{thm:normalization}, there exists such a diffeomorphism with
          \[
            \tensor{\Gamma}{^{\lambda}_{\mu\nu}}(\pull g)=0,
            \quad\text{and}\quad
            \partial_\nu\tensor{\Gamma}{^\kappa_{\lambda\mu}}(\pull g)
            =
            -\frac{1}{3} \big(\tensor{R}{^\kappa_{\lambda\mu\nu}}(\pull g) + \tensor{R}{^\kappa_{\mu\lambda\nu}}(\pull g)\big).
          \]
          Since all other variables are fixed by $\varphi$, under the covariance assumption, one can replace the variables
          $\tensor{\Gamma}{^{\lambda}_{\mu\nu}}$ by $0$ and $\tensor{\Gamma}{^{\kappa}_{\lambda\mu,\nu}}=\partial_\nu\tensor{\Gamma}{^\kappa_{\lambda\mu}}$ by the components $\tensor{R}{^{\kappa}_{\lambda\mu\nu}}$ of the curvature tensor $\bR^{g}$ (which is itself a function of $\tensor{\Gamma}{^{\kappa}_{\lambda\mu}}$ and of $\partial_{\nu}\tensor{\Gamma}{^{\kappa}_{\lambda\mu}}$).
          The variables are now
          \[
            \left(
            g,
            \bR^{g},
            \Psi,T\Psi,\Hess^{g}\Psi
            \right).
          \]

    \item
          As in the proof of \autoref{thm:first_gradient}, we can replace (using the metric $g$) the variables $\partial_{\mu}\tensor{\Psi}{^{I}}$ by the coframe $\mF_{\mat}^{\star}$, or equivalently, by the frame $\mF_{\mat}$. Finally, we can replace all the (generally covariant) tensor fields by their components in the frame $\mF_{\mat}$, which are relativistic invariants.
          The variables are now
          \[
            \left(
            \matcomponents{g^{-1}},
            \matcomponents{\bR^g},
            \Psi,
            \mF_{\mat},
            \matcomponents{\Hess^g\Psi}
            \right).
          \]

    \item
          To sum up, we have recast $L_0$ as a new Lagrangian density
          \[
            L_{0} = L(\matcomponents{g^{-1}},\matcomponents{\bR^{g}},\Psi,\mF_{\mat},\matcomponents{\Hess^{g}\Psi}),
          \]
          for some function $L$, depending on some relativistic invariants, and on the generally covariant frame $\mF_{\mat}^{\star}$.
          We conclude, as in the proof of \autoref{thm:first_gradient}, that General Covariance implies that $L$ does not depend on $\mF_{\mat}$ and thus that
          \[
            L_{0} = L(\matcomponents{g^{-1}},\matcomponents{\bR^{g}},\Psi,\matcomponents{\Hess^{g}\Psi}),
          \]
          which is obviously a relativistic invariant. This concludes the proof.
          \qedhere
  \end{enumerate}
\end{proof}

We shall now detail these components, using the orthogonal 3+1 decomposition relative to the timelike vector field $\bU$ (see~\cite[Appendix A]{KD2023}).
Note first that
\begin{equation}\label{eq:confo}
  \matcomponents{g^{-1}} =
  \begin{bmatrix}
    -1 & 0   \\
    0  & \bK
  \end{bmatrix},
\end{equation}
where $\bK$ is the conformation tensor introduced by Souriau. To express the components of the Hessian, we introduce the speed of light $c$ and some normalizations (which will be justified once a space-time structure is introduced). We thus write
\begin{equation}\label{eq:HessALK}
  \matcomponents{\Hess^g\Psi^I}
  =
  \begin{bmatrix}
    -c^{-2}A^I                      & - c^{-1}\,\pmb{L}^I \\
    - c^{-1}(\pmb{L}^I)^{\intercal} & -\pmb{H}^I
  \end{bmatrix}.
\end{equation}
Here, the matrix-by-block representation~\eqref{eq:HessALK} corresponds to the orthogonal decomposition of the Hessian tensor $\Hess^g\Psi^I$ relative to $\UU$. In tensor notation, it is written as
\[
  \Hess^g\Psi=- \frac{1}{c^{2}} \bA \otimes  \UU^{\flat} \otimes \UU^{\flat}
  + \frac{2}{c}
  \Big( \bL \, T\Psi \otimes \UU^{\flat}\Big)^{(23)}
  -  \bH(\,\cdot\,, T\Psi \,\cdot\,,T\Psi \,\cdot\,) ,
\]
where $(23)$ means the symmetrization over the second and third factors. In components, we get
\[
  (\Hess^g\Psi^I)_{\mu\nu}=- \frac{1}{c^{2}} A^I U_{\mu} U_{\nu} + \frac{1}{c} {L^I}_{J} \left( {T\Psi^{J}}_{\mu} \, U_{\nu}+{T\Psi^{J}}_{\nu}\, U_{\mu}\right)
  -  \tensor{H}{^{I}_{JK}} {T\Psi^{J}}_{\mu} {T\Psi^{K}}_{\nu},
\]
where $U_{\mu}=(\UU^{\flat})_{\mu}$.

\begin{itemize}[wide]
  \item
        The first-order tensor $\bA=(A^I)$ has the dimension of an acceleration and will be called the \textsl{material relativistic acceleration}.
        It has $3$ independent components,
        \[
          A^I
          \bydef
          - c^{2}\,(\Hess^g\Psi^I) (\UU, \UU)
          = -c^{2}\,(\nabla^{g}_{\UU}\dd\Psi^I)\cdot\UU
          ,
        \]
        and recast (since $\dd\Psi^I(\UU)=0$) as
        \[
          A^I = \dd\Psi^I\cdot(\nabla^{g}_{c\, \UU}c\, \UU),
        \]
        or, in tensorial notation, as
        \begin{equation}\label{eq:AHess}
          \bA= T\Psi\cdot \nabla^{g}_{c\, \UU}(c\, \UU).
        \end{equation}
  \item
        The second-order tensor $\bL=(\tensor{L}{^I_J})$ has the dimension of a spatial velocity gradient and will be called the \textsl{material relativistic velocity gradient}.
        It has $9$ independent components,
        \begin{equation}\label{eq:LIJmixte}
          \tensor{L}{^I_J} =(\pmb{L}^{I})_{J}
          \bydef
          - c\,(\Hess^g\Psi^I) (\UU, \EE_{J})
          = -c\,\nabla^{g}_{\EE_{J}}(\dd\Psi^I)\cdot\UU,
        \end{equation}
        and recast (since $\dd\Psi^I(\UU)=0$) as
        \[
          \tensor{L}{^I_J}
          =
          \dd\Psi^I\cdot\nabla^{g}_{\EE_{J}}(c\, \UU),
        \]
        or, in tensorial notation, as
        \begin{equation}\label{eq:LHess}
          \bL=T\Psi\,( \nabla^{g} c\, \UU) \, (T\Psi)_{\UU^{\perp}}^{-1}.
        \end{equation}
  \item
        The third-order tensor $\bH=(\tensor{H}{^{I}_{JK}})$ has the dimension of a strain gradient and will be called the \textsl{material relativistic strain gradient}.
        It satisfies the index symmetry $\tensor{H}{^{I}_{JK}}=\tensor{H}{^{I}_{KJ}}$ and hence has $18$ independent components,
        \begin{equation}\label{eq:HIJKmixte}
          \tensor{H}{^{I}_{JK}}=(\pmb{H}^{I})_{JK}
          \bydef
          - (\Hess^g\Psi^I) (\EE_{J}, \EE_{K})
          = -\nabla^{g}_{\EE_{K}}(\dd\Psi^I)\cdot\EE_{J},
        \end{equation}
        and recast (since $\dd\Psi^I(\EE_{J})=\delta^{I}_{J}$) as
        \[
          \tensor{H}{^{I}_{JK}}
          =
          \dd\Psi^I\cdot\nabla^{g}_{\EE_{K}}(\EE_{J})
          =
          \dd\Psi^I\cdot\nabla^{g}_{\EE_{J}}(\EE_{K}),
        \]
        or, in tensorial notation, as
        \begin{equation}\label{eq:HHess}
          \bH = -(\nabla^{g}T\Psi) \left(\,\cdot\,, (T\Psi)_{\UU^{\perp}}^{-1} \,\cdot\,,  (T\Psi)_{\UU^{\perp}}^{-1}\,\cdot\, \right)=T\Psi \cdot\nabla^{g}_{\EE_{K}}(\EE_{J})\otimes \dd X^{J} \otimes  \dd X^{K}.
        \end{equation}
\end{itemize}

We will not use the components of the curvature tensor $\bR^{g}$ for applications in this work. However, the data contained in $\matcomponents{\bR^{g}}$ will be reorganized into three new relativistic gravitation/matter coupling tensors. These tensors $\bM_{p}$ represent new coupling invariants for gravitation/matter and could be useful at the astrophysical scale, for the study of neutron stars, of black holes and of dark energy/matter. They vanish for a flat spacetime such as the Minkowski spacetime of Special Relativity. These relativistic invariants are defined as follows.

\begin{itemize}
  \item
        The fourth-order tensor $\bM_4=(\tensor{M}{^I_{JKL}})$ is given, in components, by
        \[
          \tensor{M}{^I_{JKL}} \bydef \bR^{g}(\dd \Psi^{I}, \EE_{J}, \EE_{K}, \EE_{L}) = \tensor{R}{^{\kappa}_{\lambda\mu\nu}}
          {T \Psi^I}_{\kappa} \, E_J^\lambda\, E_K^\mu\, E_L^\nu,
        \]
        or, in tensorial notation, by
        \begin{equation}\label{eq:M4}
          \bM_{4} = \bR^{g}(T\Psi^{\star} \, \cdot\, , (T\Psi)_{\UU^{\perp}}^{-1} \,\cdot\,, (T\Psi)_{\UU^{\perp}}^{-1} \,\cdot\,, (T\Psi)_{\UU^{\perp}}^{-1} \,\cdot\,  ).
        \end{equation}
  \item
        The third-order tensor $\bM_3=(\tensor{M}{^I_{JK}})$ is given, in components, by
        \[
          \tensor{M}{^I_{JK}} \bydef c\, \bR^{g}(\dd \Psi^I, \EE_J, \EE_K,  \UU) = c\, \tensor{R}{^{\kappa}_{\lambda\mu\nu}}
          {T \Psi^I}_{\kappa} \, E_J^\lambda\,E_K^\mu\, U^\nu,
        \]
        or, in tensorial notation, by
        \begin{equation}\label{eq:M3}
          \bM_{3}=c\, \bR^{g}(T\Psi^{\star}  \, \cdot\, , (T\Psi)_{\UU^{\perp}}^{-1} \,\cdot\,, (T\Psi)_{\UU^{\perp}}^{-1} \,\cdot\,, \UU ).
        \end{equation}
  \item
        The second-order tensor $\bM_2=(\tensor{M}{^I_{J}})$ is given, in components, by
        \[
          \tensor{M}{^I_{J}} \bydef \bR^{g}(\dd \Psi^I, c\, \UU, \EE_J, c\, \UU) = c^{2}\, \tensor{R}{^{\kappa}_{\lambda\mu\nu}}
          {T \Psi^I}_{\kappa} \, U^\lambda\, E_J^\mu \,  U^\nu,
        \]
        or, in tensorial notation, by
        \begin{equation}\label{eq:M2}
          \bM_{2}=c^{2}\, \bR^{g}(T\Psi^{\star} \, \cdot\, , \UU , (T\Psi)_{\UU^{\perp}}^{-1} \,\cdot\,, \UU ).
        \end{equation}
\end{itemize}

\begin{rem}
  The Riemann tensor $g\bR^{g} = (R_{\kappa\lambda\mu\nu})$ having the following symmetries
  \[
    R_{\kappa\lambda\mu\nu} = R_{\mu\nu\kappa\lambda} = -R_{\lambda\kappa\mu\nu} = - R_{\kappa\lambda\nu\mu},
    \quad \text{and} \quad
    R_{\kappa\lambda\mu\nu} + R_{\lambda\mu\kappa\nu} + R_{\mu\kappa\lambda\nu}  = 0,
  \]
  and $20$ independent components in dimension $4$, there are thus $6$ independent components for $\bM_{4}$, $8$ for $\bM_{3}$, and $6$ for $\bM_{2}$.
\end{rem}

Reformulating \autoref{thm:second_gradient} using these new secondary variables, we get the following corollary.

\begin{cor}\label{cor:second_gradient}
  If the Lagrangian
  \[
    \mL^{\text{matter}}[g, \Psi]
    =
    \int
    L_{0}(
    \tensor{g}{_{\mu\nu}}, \partial_{\nu}\tensor{g}{_{\lambda\mu}},\partial_{\mu}\partial_{\nu}\tensor{g}{_{\kappa\lambda}},
    \tensor{\Psi}{^{I}}, \partial_{\mu}\tensor{\Psi}{^{I}},\partial_{\mu}\partial_{\nu}\tensor{\Psi}{^{I}}
    )\vol_{g}
  \]
  is generally covariant under the group of orientation-preserving diffeomorphisms, then, its Lagrangian density $L_{0}$ can be written as a function
  \[
    L_{0} = L\big(\Psi, \bK, \bM_{2}, \bM_{3}, \bM_{4}, \bA, \bL, \bH\big)
  \]
  of the vector-valued relativistic invariants $\Psi, \bK, \bM_{2}, \bM_{3}, \bM_{4}, \bA,\bL, \bH $.
\end{cor}

\begin{rem}
  Note that, in contrast with \autoref{thm:first_gradient}, \autoref{thm:second_gradient} does not extend to covariance under the full group of general diffeomorphisms. Components that involve the timelike vector field $\UU$ an even number of times are invariant under the group of general diffeomorphisms, but components that involve it an odd number of times change sign when the orientation is reversed. For instance, $\bA$ is invariant under the full group of general diffeomorphisms like $\bH$ and the conformation $\bK$. On the contrary, $\bL$ is not. Whether there is a link with the fact that the Galilean limits of $\bK$ and $\bH$ are objective, whereas the limit of $\bL$ is not (but the limit of its symmetric part is), remains an open question for us.
\end{rem}

\begin{rem}\label{rem:alternative-invariants}
  In the spirit of \autoref{rem:PushCometric}, one can formulate alternative relativistic invariants by using the contravariant forms $(\Hess^g \Psi)^{\sharp}$ and $(\bR_{g})^{\sharp}$ of the tensor fields $\Hess^g \Psi$ and $\bR_{g}$. These alternative invariants are the components of $(\Hess^g \Psi)^{\sharp}$ and $(\bR_{g})^{\sharp}$---up to normalization factors--- in the rest coframe $\mF_{\mathrm{rest}}^{\star} = (-\UU^{\flat}, \dd \Psi^{I})$. For instance, instead of the relativistic invariants $\tensor{H}{^{I}_{JK}}$ (defined by \eqref{eq:HIJKmixte}), one may introduce
  \[
    H^{IJK} \bydef -(\Hess^g \Psi^{I})^{\sharp} (\dd \Psi^{J},\dd \Psi^{K}) = -\Hess^g \Psi^{I}(\grad^{g} \Psi^{J},\grad^{g} \Psi^{K}).
  \]
  But since
  $\grad^{g} \Psi^{J} = K^{IP} \EE_{P}$,
  we get
  \[
    H^{IJK} = K^{JP}K^{KQ}H^{I}_{PQ} = (\bH^{\sharp_{\bK^{-1}}})^{IJK},
  \]
  where the inverse of the conformation, $\bK^{-1}$, is used as a metric to raise indices on the vector space $V$. The $H^{IJK}$ are the components of the \textsl{contravariant material relativistic strain gradient},
  \[
    \bH^{\sharp_{\bK^{-1}}}=(\Hess^g \Psi)^{\sharp} (\cdot, T\Psi^{\star} \cdot, T\Psi^{\star} \cdot).
  \]
  Similarly, instead of the relativistic invariants $\tensor{L}{^{I}_{J}}$ (defined by \eqref{eq:LIJmixte}) one has
  \[
    L^{IJ} \bydef c\, (\Hess^g \Psi^{I})^{\sharp} (-\UU^{\flat},\dd \Psi^{J}) = -c\, \Hess^g \Psi^{I}(\UU,\grad^{g} \Psi^{J}) =  \tensor{L}{^{I}_{P}}K^{PJ},
  \]
  which are the components of the  \textsl{contravariant material relativistic velocity gradient}
  \[
    \bL^{\sharp_{\bK^{-1}}}=\bL\cdot \bK.
  \]
  One also has
  \begin{align*}
    \tensor{M}{^{IJKL}} & \bydef \left(\bR^{g}\right)^{\sharp} (\dd \Psi^{I}, \dd \Psi^{J}, \dd \Psi^{K}, \dd \Psi^{L}) = K^{JP} K^{KQ} K^{LR} \tensor{M}{^I_{PQR}},
    \\
    \tensor{M}{^{IJK}}  & \bydef - c\,\left(\bR^{g}\right)^{\sharp} (\dd \Psi^{I}, \dd \Psi^{J}, \dd \Psi^{K}, -\UU^{\flat}) =  K^{JP} K^{KQ} \tensor{M}{^I_{PQ}},
    \\
    \tensor{M}{^{IJ}}   & \bydef c^{2}\,\left(\bR^{g}\right)^{\sharp} (\dd \Psi^{I}, \UU^{\flat}, \dd \Psi^{K}, \UU^{\flat}) = K^{JP} \tensor{M}{^I_{P}}.
  \end{align*}
  which are the components of $(\bM_{4})^{\sharp_{\bK^{-1}}}$, $(\bM_{3})^{\sharp_{\bK^{-1}}}$ and $(\bM_{2})^{\sharp_{\bK^{-1}}}=\bM_{2}\cdot \bK$, respectively.
\end{rem}

\subsection{Gradient Fluids}
\label{sec:gradient_fluids}

The classical (3D) gradient theory for fluids~\cite{FG2006,Gou2019} involves only the secondary variable $\rho$ (the fluid mass density) and considers its gradient $\grad \rho$, or $\grad \ln \rho$, as the main variable for constitutive modeling. In the present relativistic framework, we shall consider, as the fundamental secondary variable, the rest mass density
\begin{equation}\label{eq:rhorHere}
  \rho_{r}(g, \Psi, T\Psi)
  =
  (\rho_{\circ} \circ \Psi)
  \sqrt{\det(K^{IJ})},
  \quad \text{where} \quad
  \bK=T\Psi\, g^{-1} \, T\Psi^{\star},
\end{equation}
and $\rho_{\circ}$ is the mass density on the body $\body$ (a parameter). The function $\rho_{r}$ is only defined on the World tube $\mW$ and can be used only to describe \emph{interior solutions}. The vector field $\grad^{g}( \ln\rho_r)$ is a function of the metric $g$ and the second-order jet of the matter field $\Psi$. Note that $\rho_{r}$ is a relativistic invariant but $\grad^{g}( \ln\rho_r)$ is not; it is only a relativistic covariant, meaning that
\[
  \grad^{g}( \ln\rho_r)(\varphi^{*}g, \varphi^{*}\Psi, \varphi^{*}T\Psi) = \varphi^{*} \left(\grad^{g}( \ln\rho_r)(g, \Psi, T\Psi)\right),
\]
for every (local) diffeomorphism $\varphi$. However, the components of this vector field in the matter frame $\mF_{\mathrm{rest}}$
\[
  h^{0} \bydef -\UU^{\flat} \cdot \grad^{g}(\ln\rho_r) = - \dd(\ln\rho_r) \cdot \UU,
  \qquad
  h^{I} \bydef \dd\Psi^{I} \cdot \grad^{g}( \ln\rho_r) = \dd(\ln\rho_r) \cdot \grad^{g} \Psi^{I}
\]
are relativistic invariants. This means that one can formulate a generally covariant gradient theory for relativistic fluids using the following matter Lagrangian density
\[
  L = \rho_{r} c^{2} + E(\rho_{r}, h^{0}, h^{I}),
\]
where the pure mass term $\rho_{r} c^{2}$ models dust \cite{Sou1958,Sou1964}.

By Corollary~\ref{cor:second_gradient}, the Lagrangian density $L$ should be expressible using the relativistic invariants $\Psi$, $\bK$, $\bA$, $\bL$, $\bH$. This is indeed the case, and we summarize it in the following lemma.

\begin{lem}\label{lem:gradient-fluids}
  Let $(X^{I})$ be the canonical coordinate system on $V=\RR^{3}$.
  Given a mass measure
  \[
    \mu=\rho_{\circ} \dd X^{1}\wedge \dd X^{2} \wedge \dd X^{3}
  \]
  on the body $\body$, we have
  \[
    h^{0} = \frac{1}{c} \tr \bL \quad \text{and} \quad \pmb{h} = \left(\dd (\ln \rho_{\circ}) \circ \Psi\right) \cdot \bK -  \left(\tr_{12} \bH\right)\cdot \bK,
  \]
  where $\pmb{h} = (h^{I})$, $\bK=(K^{IJ})$ is the conformation and $\tr_{12} \bH=(\tensor{H}{^{J}_{JI}})$.
\end{lem}

\begin{rem}
  For an homogeneous fluid ($\rho_{\circ} =  \text{constant}$), we get simply
  \[
    \pmb{h} = -\left(\tr_{12} \bH\right) \cdot \bK , \qquad h^{I}= -\tensor{H}{^{J}_{JL}}K^{LI}.
  \]
\end{rem}

\begin{proof}[Proof of Lemma~\ref{lem:gradient-fluids}]
  By~\eqref{eq:rhorHere} , we have
  \[
    \dd(\ln\rho_r) = \dd(\ln\rho_{\circ}\circ\Psi) + \frac{1}{2} \dd(\ln\det(K^{IJ})).
  \]
  But
  \[
    \dd(\ln\det(K^{IJ})) = (K^{-1})_{IJ}\dd(K^{IJ}).
  \]
  Now, for any vector field $\XX$, one has
  \begin{align*}
    \dd K^{IJ}\cdot\XX & = \nabla_{\XX}^{g}\big(g^{-1}(\dd \Psi^{I}, \dd \Psi^{J})\big)                                                         \\
                       & = g^{-1}\big(\nabla^{g}_{\XX} \dd\Psi^{I}, \dd\Psi^{J}\big) + g^{-1}\big(\dd\Psi^{I}, \nabla^{g}_{\XX}\dd\Psi^{J}\big) \\
                       & = \Hess^g\Psi^I(\XX, \grad^{g}\Psi^{J}) + \Hess^g\Psi^J(\XX,\grad^{g}\Psi^{I}).
  \end{align*}
  But
  \[
    \XX = X^{0}\UU + X^{L} \EE_{L}, \quad \grad^{g}\Psi^{I} = \bK^{IM} \EE_{M},
  \]
  with
  \[
    \Hess^g\Psi^I(\UU,\EE_{M}) = -\frac{1}{c} \, \tensor{L}{^{I}_{M}}, \qquad \Hess^g\Psi^I(\EE_{L},\EE_{M}) = -\tensor{H}{^{I}_{LM}},
  \]
  and thus
  \[
    \Hess^g\Psi^I(\XX, \grad^{g}\Psi^{J}) = -\frac{1}{c} \, X^{0} K^{JM} \tensor{L}{^{I}_{M}} -  X^{L} K^{JM}\tensor{H}{^{I}_{LM}}.
  \]
  We get therefore
  \begin{align*}
    \dd(\ln\det(K^{IJ}))\cdot\XX & = (K^{-1})_{IJ}\dd(K^{IJ})\cdot\XX                                                                                          \\
                                 & = (K^{-1})_{IJ} \left( -\frac{1}{c} \, X^{0} K^{JM} \tensor{L}{^{I}_{M}} -  X^{L} K^{JM}\tensor{H}{^{I}_{LM}} \right)       \\
                                 & \quad + (K^{-1})_{IJ} \left( -\frac{1}{c} \, X^{0} K^{IM} \tensor{L}{^{J}_{M}} -  X^{L} K^{IM}\tensor{H}{^{J}_{LM}} \right) \\
                                 & = -2 \left( \frac{1}{c} \, X^{0} \tensor{L}{^{M}_{M}} + X^{L} \tensor{H}{^{M}_{LM}} \right),
  \end{align*}
  from which we deduce that
  \[
    \frac{1}{2} \dd(\ln\det(K^{IJ})).\UU = -\frac{1}{c} \tensor{L}{^{M}_{M}},
  \]
  and
  \[
    \frac{1}{2} \dd(\ln\det(K^{IJ})). \grad^{g}\Psi^{N} = -K^{NL}\tensor{H}{^{M}_{LM}} = - \tensor{H}{^{M}_{ML}}K^{LN}.
  \]

  By the way, we have
  \[
    \dd(\ln\rho_{\circ}\circ\Psi).\UU = 0, \quad \text{and} \quad \dd(\ln\rho_{\circ}\circ\Psi).\grad^{g}\Psi^{N} = (\dd \ln\rho_{\circ})_{L}K^{LN}.
  \]
  We get thus finally
  \[
    h^{0} =  \frac{1}{c} \tr \bL \quad \text{and} \quad \pmb{h} = h^{N}\EE_{N} = \left((\dd \ln \rho_{\circ} )\circ \Psi\right) \cdot\bK  - \left( \tr_{12} \bH\right)\cdot \bK.
    \qedhere
  \]
\end{proof}

\begin{rem}
  Different relativistic invariants can be obtained if one starts with $\dd (\ln \rho_{r})$ rather than $\grad^{g}( \ln\rho_r)$. In that case, one gets
  \[
    k_{0} = -\frac{1}{c} \tr \bL, \quad \text{and} \quad \pmb{k} = (\dd \ln \rho_{\circ} )\circ \Psi  - \tr_{12} \bH.
  \]
\end{rem}

\section{Time, Space and observer frames}
\label{sec:observer-frames}
\subsection{Observer}
In General Relativity, observing the motion of matter requires the introduction of an arbitrary \textsl{time function} $\hat{t}$ on the Universe $(\univ,g)$, that is a smooth (local) submersion
\[
  \hat{t}\colon\univ\to\RR,
\]
defined on $\univ$, whose gradient is everywhere timelike. This function defines a foliation (called a \emph{spacetime structure} or $(3+1)$-structure \cite{ADM1962,Yor1979,Gou2012}) of $\mM$ by spacelike hypersurfaces
\[
  \mathcal{E}_{t} := \set{m\in\mM;\;\hat{t}(m)=t}.
\]
Let $\bN$ be the future-oriented (\textit{i.e.}, $g(\bN,\grad^{g}\hat{t})<0$) normal vector field to $\lab_{t}$. Then, the orthogonal decomposition relative to $\bN$ of the metric $g$ is given by
\[
  g = -\bN^{\flat} \otimes \bN^{\flat}+\gE,
  \quad\text{with}\quad
  \gE\,\bN=0.
\]
The degenerate metric $\mathfrak{g}$, when restricted to each hypersurface $\mathcal{E}_{t}$, is a Riemannian metric and coincides with the restriction of the Lorentzian metric $g$ to $\mathcal{E}_{t}$~\cite[Page 6]{Lic1955}.

We shall assume furthermore that $g(\UU,\grad^{g}\hat{t})<0$, where the timelike vector field $\UU$ was defined in \autoref{subsec:matter_field} and is related to matter. This means that $\grad^{g}\hat{t}$ and $\UU$ define the same time-orientation. This allows us to define the generalized \emph{Lorentz factor} as the scalar field \cite{Gou2012}
\[
  \gamma := -g(\UU,\bN),
\]
and the \textsl{spatial velocity} $\uu$, defined by the orthogonal decomposition of $\UU$ relative to $\bN$ (see~\cite{Sou1958,Sou1964,Gou2012,KD2023}), given by
\begin{equation}\label{eq:Uu}
  \UU
  =
  \gamma\left(
  \bN
  +
  \frac{\uu}{c}
  \right),
  \qquad\text{with}\qquad
  g(\uu, \bN)=0.
\end{equation}
Since $g(\UU,\UU)=g(\bN,\bN)=-1$, one recovers the familiar expression
\[
  \gamma=\frac{1}{\sqrt{1-\dfrac{\norm{\uu}_{g}^{2}}{c^{2}}}}.
\]

Let us now consider the spacelike hypersurfaces of the World tube
\[
  \Omega_{t}
  \bydef
  \mW\cap\lab_t,
\]
that play the same role as the configuration manifolds in Classical Continuum Mechanics. Note that whereas the restriction $\Psi_{t}\colon\Omega_{t}\to\body$ is not assumed (\textit{a priori}) bijective, its differential $T\Psi_{t}$ is always a linear isomorphism. Since $\lab_t$ is orthogonal to $\bN$, at a point $m$ of $\Omega_{t}$ one has the orthogonal decomposition
\[
  T\Psi
  =
  - (T\Psi\cdot \bN) \otimes \bN^\flat
  +
  T\Psi_{t},
\]
where $T\Psi_{t}$ is the restriction of $T\Psi$ to the spacelike vector space $T\Omega_{t}$. We denote by $\bF_m$ the linear map $T_{\Psi(m)}\body\to T_{m}\Omega_{t}$ (the \textsl{deformation gradient}, see~\cite{KD2023}) defined by
\[
  \bF_m
  \bydef
  (T_m\Psi_t)^{-1}.
\]

\begin{rem}
  The deformation gradient $\bF$ corresponds to the right inverse $(T\Psi)_{\bN^{\perp}}^{-1}$, which is characterized by
  \[
    T\Psi\cdot (T\Psi)_{\bN^{\perp}}^{-1} = \id_{V}, \quad \text{and} \quad \bN^{\flat} \cdot (T\Psi)_{\bN^{\perp}}^{-1} = 0.
  \]
  Similarly, the right inverse $(T\Psi)_{\UU^{\perp}}^{-1}$ of the linear tangent map $T\Psi$ is the unique solution of
  \[
    T\Psi\cdot (T\Psi)_{\UU^{\perp}}^{-1} = \id_{V}, \quad \text{and} \quad \UU^{\flat} \cdot (T\Psi)_{\UU^{\perp}}^{-1} = 0.
  \]
  The relation between these two right inverses is then given by
  \[
    (T\Psi)_{\UU^{\perp}}^{-1}=\UU \otimes (\UU^{\flat}\cdot \bF) + \bF.
  \]
\end{rem}

Using the fact that
\[
  T\Psi\cdot\UU = \gamma \left( T\Psi\cdot\bN + \frac{1}{c}\, T\Psi\cdot\uu \right) = \gamma \left( T\Psi\cdot\bN + \frac{1}{c}\, T\Psi_{t} \cdot\uu \right) = 0,
\]
since $\uu$ is tangent to $\Omega_t$, we get
\[
  T\Psi\cdot \bN  = -\frac{1}{c}\, \bF^{-1}\uu.
\]
The orthogonal decomposition of $T\Psi$ relative to $\bN$ can thus be recast as
\[
  T\Psi = \frac{1}{c} (\bF^{-1}\uu) \otimes \bN^\flat + \bF^{-1},
  \qquad
  \bF^{-1}\bN=0.
\]

The column vectors $\FF_J\bydef\bF\cdot \partial/\partial X^{J}$ of $\bF$ in the canonical basis of $\RR^3$ are characterized by the properties:
\begin{equation}\label{eq:F_I}
  \bN^\flat(\FF_J)=0
  \quad\text{and}\quad
  \dd\Psi^I_{t}(\FF_J)=\tensor{\delta}{^I_J}.
\end{equation}
We shall define the \textsl{observer frame} as
\[
  \mF_\mathrm{obs}
  \bydef
  (\bN,\FF_1,\FF_2,\FF_3).
\]
Beware that the vector fields in this frame are \emph{not} generally covariant of $g$ and $\Psi$.

\begin{rem}
  The frame $\mF_{\mathrm{rest}}$ introduced in Theorem~\ref{thm:rest-frame} is \emph{linked to the matter} and is thus the \emph{rest frame} of the matter (meaning that the spatial velocity $\uu = 0$ and the Lorentz factor $\gamma = 1$ in this frame). In General Relativity, any frame
  \[
    \mF \bydef \left(\bN,\ee_{i} \right), \quad \text{with} \quad g(\bN,\ee_{i}) = 0,
  \]
  respecting the orthogonal $(3+1)$-decomposition induced by a time function is called an \textsl{Eulerian frame}~\cite{Gou2012}. Our observer frame $\mF_\mathrm{obs}$ is thus a particular choice of Eulerian frame.
\end{rem}

\subsection{Rewriting the invariants in the observer frame}

Rewriting the vector-valued invariants $\bA$, $\bL$ and $\bH$ using the components of the Hessian $\Hess^g\Psi^{I}$ in the observer frame $(\bN, \FF_{1}, \FF_{2}, \FF_{3})$ will give expressions that are more mechanically meaningful than the initial expressions \eqref{eq:AHess}, \eqref{eq:LHess} and \eqref{eq:HHess}. These expressions will be used to compute their classical limit, and to recover the invariants of classical second gradient elasticity. However, to obtain a compact statement and avoid overly long formulas, we first keep the vector field $\UU$ and express the invariants in the frame $(\UU, \FF_{1}, \FF_{2}, \FF_{3})$ rather than $(\bN, \FF_{1}, \FF_{2}, \FF_{3})$, as illustrated in the following result.

\begin{lem}\label{lem:hessian}
  We have
  \[
    \left\{
    \begin{aligned}
       & A^{I} = \dd\Psi^I\cdot\nabla^{g}_{c\UU}(c\UU),
      \\
       & \tensor{L}{^{I}_{J}} = \dd\Psi^I\cdot\nabla^{g}_{\FF_{J}}(c\UU) + \frac{1}{c} A^{I}\UU^\flat(\FF_J),
      \\
       & \tensor{H}{^{I}_{JK}} = \dd\Psi^I\cdot\nabla^{g}_{\FF_{K}}(\FF_{J}) +  \frac{1}{c}\left(\tensor{L}{^{I}_{J}}\, \UU^\flat(\FF_K)+\tensor{L}{^{I}_{K}}\, \UU^\flat(\FF_J) \right)+\frac{1}{c^2} A^{I}\,\UU^\flat(\FF_J)\,\UU^\flat(\FF_K),
    \end{aligned}
    \right.
  \]
  or, in tensorial notation
  \[
    \left\{
    \begin{aligned}
      \bA & = T\Psi\cdot\nabla^{g}_{c\UU}(c\UU),
      \\
      \bL & = T\Psi \cdot (\nabla^{g} c\UU) \cdot \bF + \frac{1}{c} \bA \otimes (\UU^\flat \bF),
      \\
      \bH & = (T\Psi\cdot\nabla^{g}_{\FF_{K}}(\FF_{J})) \otimes \dd X^{J} \otimes \dd X^{K} + \frac{2}{c} \big(\bL \otimes (\UU^\flat\bF)\big)^{(23)} + \frac{1}{c^{2}} \bA \otimes (\UU^\flat\bF) \otimes (\UU^\flat\bF),
    \end{aligned}
    \right.
  \]
  where $(\cdot)^{(23)}$ means the symmetrization over the second and third factors.
\end{lem}

\begin{proof}
  From \eqref{eq:E_I} and \eqref{eq:F_I}, and the fact that $\dd\Psi^{I} \cdot \bF = \dd X^{I}$, one deduces that
  \begin{equation}\label{eq:EJ-FJ}
    \EE_J = \FF_J + \UU^\flat(\bF_J)\UU,
  \end{equation}
  and one has
  \[
    \dd\Psi^{I}\cdot\UU = 0 \quad \text{and} \quad \dd\Psi^{I}\cdot\EE_{J} = \dd\Psi^{I}\cdot\FF_{J}= \delta^{I}_{J}.
  \]
  Now, using the Leibniz rule
  \[
    (\Hess^g\Psi^{I}) (\XX, \YY) = \nabla^{g}_{\XX} \dd\Psi^{I} \cdot \YY = \dd (\dd\Psi^{I}\cdot \XX)\cdot \YY - \dd\Psi^{I}\cdot \nabla^{g}_{\XX}\YY,
  \]
  and the fact that the Hessian is symmetric (because the Levi-Civita connection is torsionless), one gets
  \begin{align*}
     & (\Hess^g\Psi^{I}) (c\UU,c\UU) = - \dd\Psi^I\cdot\nabla^{g}_{c\UU}(c\UU),                                                                                                   \\
     & (\Hess^g\Psi^{I}) (c\UU,\FF_{J}) = (\Hess^g\Psi^{I}) (\FF_{J},c\UU) = - \dd\Psi^I\cdot\nabla^{g}_{c\UU}(\FF_{J}) = - \dd\Psi^I\cdot\nabla^{g}_{\FF_{J}}(c\UU),             \\
     & (\Hess^g\Psi^{I}) (\FF_{J},\FF_{K}) = (\Hess^g\Psi^{I}) (\FF_{K},\FF_{J}) = - \dd\Psi^I\cdot\nabla^{g}_{\FF_{J}}(\FF_{K}) = - \dd\Psi^I\cdot\nabla^{g}_{\FF_{K}}(\FF_{J}).
  \end{align*}
  We therefore get, using \eqref{eq:EJ-FJ} and the bilinearity of the Hessian:
  \begin{itemize}
    \item Firstly:
          \[
            A^{I} = - c^{2}\,(\Hess^g\Psi^{I}) (\UU, \UU) = \dd\Psi^I\cdot\nabla^{g}_{c\UU}(c\UU),
          \]

    \item Secondly:
          \begin{align*}
            \tensor{L}{^{I}_{J}} & = - c\,(\Hess^g\Psi^{I}) (\UU, \EE_J)
            \\
                                 & = - c\,(\Hess^g\Psi^{I}) (\UU, \FF_J+ \UU^\flat(\FF_J)\,\UU)
            \\
                                 & = - c\,(\Hess^g\Psi^{I}) (\UU, \FF_J) -c  \,(\Hess^g\Psi^{I}) (\UU,\UU) \, \UU^\flat(\FF_J)
            \\
                                 & = \dd\Psi^I\cdot\nabla^{g}_{\FF_{J}}(c\UU) +\frac{1}{c}\, A^{I}\, \UU^\flat(\FF_J).
          \end{align*}
    \item Thirdly:
          \begin{align*}
            \tensor{H}{^{I}_{JK}} & = - (\Hess^{g}\Psi^I)(\EE_{J}, \EE_{K})
            \\
                                  & = - (\Hess^{g}\Psi^I)(\FF_J+ \UU^\flat(\FF_J)\,\UU, \FF_K+ \UU^\flat(\FF_K)\,\UU)
            \\
                                  & = - (\Hess^{g}\Psi^I)(\FF_J, \FF_K) - (\Hess^{g}\Psi^I)(\FF_J, \UU)\, \UU^\flat(\FF_K)
            \\
                                  & \quad - (\Hess^{g}\Psi^I)(\UU, \FF_K)\, \UU^\flat(\FF_J)
            - (\Hess^{g}\Psi^I)(\UU, \UU) \, \UU^\flat(\FF_J)\, \UU^\flat(\FF_K)
            \\
                                  & = \dd\Psi^I\cdot\nabla^{g}_{\FF_{J}}(\FF_{K}) + \frac{1}{c}\,\tensor{L}{^{I}_{J}}\, \UU^\flat(\FF_K)+ \frac{1}{c}\,\tensor{L}{^{I}_{K}}\, \UU^\flat(\FF_J) +\frac{1}{c^2}\, A^{I}\,\UU^\flat(\FF_J)\,\UU^\flat(\FF_K).
          \end{align*}
  \end{itemize}
\end{proof}

\subsection{Second Gradient Elasticity in Minkowski spacetime}

We will now specify the general theory in a familiar simple setting. Since the presence of the matter under study in the laboratory hardly affects the metric $g$ of the Universe, compared to the presence of the Sun, the Earth and other planets (\textsl{passive matter} assumption), we can consider that $g$ is now fixed and choose, among the numerous spacetimes encountered in General Relativity and listed in the catalog~\cite{MG2009}, those describing vacuum solutions of Einstein's equations. In this paper, we have chosen to work in the \textsl{Minkowski spacetime} $(\univ,\eta)$, the framework of Special Relativity. It is a flat pseudo-Riemannian manifold, diffeomorphic to $\RR^{4}$ with canonical coordinates $(x^{0}=ct, x^{i})$ and endowed with the static metric
\[
  \eta = -c^{2} \dd t^{2}+ \gE,
  \quad \text{where} \quad
  \gE = \delta_{ij} \dd x^{i} \dd x^{j}.
\]
The time function is here $\hat{t}(m)=t=x^{0}/c$ and accordingly, one has
\[
  \bN = \frac{1}{c}\pd{t}{}, \quad \text{and} \quad \bN^{\flat}=- c \dd t.
\]
Note that in this simple case $\nabla^\eta\bN=0$. As a consequence, for any vector field $\XX$ tangent to $\lab_{t}$, one has
\[
  0 = \dd g(\XX,\bN) = g(\nabla^\eta\XX,\bN) + g(\nabla^\eta\bN,\XX),
\]
and thus
\[
  g(\nabla^\eta\XX,\bN) = -g(\nabla^\eta\bN,\XX) = 0.
\]
For such vector fields, at a point of $\Omega_{t}$, one has therefore
\[
  T\Psi \cdot \nabla^\eta\XX = T\Psi_{t} \cdot \nabla^\eta\XX = \bF^{-1} \cdot \nabla^\eta\XX,
\]
where $\bF \bydef T\Psi_{t}^{-1} \colon T\body\to T\Omega_{t}$. If $\XX$ and $\YY$ are both tangent to $\lab_{t}$, one has furthermore
\[
  \nabla^{\eta}_\XX\YY = \nabla^{\gE}_\XX\YY,
\]
\textit{i.e.}, the hypersurface $\lab_t$ of $(\univ,g)$ has no extrinsic curvature. Here, $ (\lab_t,\gE)$ is a three-dimensional Euclidean affine space. Finally, the four-velocity and the Lorentz factor have the following expressions
\[
  c\,\UU = \gamma\left(\pd{t}{}+\uu\right),
  \qquad
  \gamma=\frac{1}{\sqrt{1-\dfrac{\norm{\uu}_{\gE}^{2}}{c^{2}}}}.
\]
and one has
\begin{equation}\label{eq:UbF}
  \UU^\flat\bF = \frac{\gamma}{c} \uu^\flat \bF.
\end{equation}

\begin{rem}\label{rem:material-derivative}
  One can observe that $\gamma^{-1}\nabla^\eta_{ c\, \UU} \XX$ corresponds here to the (relativistic) particle derivative, which we denote with a dot,
  \[
    \dot {\XX} \bydef \gamma^{-1}\nabla^{\eta}_{ c\, \UU}\XX = \partial_{t} \XX + \nabla^{\gE}_{\uu} \XX,
  \]
  and that if $\XX$ is tangent to $\lab_{t}$, then so is $\dot {\XX}$.
\end{rem}

Since the curvature tensor vanishes, the invariants $\bM_{2}$, $\bM_{3}$ and $\bM_{4}$ vanish and we need only make explicit the remaining vector-valued relativistic invariants $\bK$, $\bA$, $\bL$ and $\bH$.

\begin{thm}\label{thm:limits}
  In the Minkowski spacetime:
  \begin{itemize}
    \item The \emph{conformation} $\bK = (K^{IJ})$ is given by
          \[
            \bK = \bF^{-1} \left(\gE^{-1}-\frac{1}{c^{2}}{\uu\otimes\uu}\right)\bF^{-\star},
          \]
          which corresponds to the formula derived in \cite{Sou1964}.

    \item The \emph{material relativistic acceleration} $\bA = (A^{I})$ is given by
          \[
            \bA =  \bF^{-1}\,\ba, \quad \text{where} \quad \ba := \gamma^{2}\, \dot \uu = \gamma^{2}\left(\partial_{t} \uu + \nabla^{\gE}_{\uu} \uu\right).
          \]
          Observe that $\bA$ is the pullback to the body $\body$ of the \emph{spatial relativistic acceleration} $\ba$, defined on $\Omega_{t}$.

    \item The \emph{material relativistic velocity gradient} $\bL = (L^{IJ})$ is given by
          \[
            \bL = \bF^{-1} \cdot \left(\gamma \, \nabla^{\gE} \uu + \frac{\gamma}{c^{2}} \, \ba \otimes \uu^\flat \right) \cdot \bF
          \]
          Observe that $\bL$ is the pullback to the body $\body$ of the spatial mixed tensor
          \[
            \gamma \, \nabla^{\gE} \uu + \frac{\gamma}{c^{2}} \, \ba \otimes \uu^\flat,
          \]
          defined on $\Omega_{t}$.

    \item The \emph{material relativistic strain gradient} $\bH = (\tensor{H}{^{I}_{JK}})$ is given by
          \begin{equation}\label{eq:HMink}
            \bH = \bF^{-1} \cdot \nabla^{\gE} \bF \cdot \bF + \frac{2\gamma}{c^{2}} \big(\bL \otimes (\uu^\flat\bF)\big)^{(23)} + \frac{\gamma^{2}}{c^{4}} \bA \otimes (\uu^\flat\bF) \otimes (\uu^\flat\bF).
          \end{equation}
          where $(\cdot)^{(23)}$ means the symmetrization over the second and third factors, and
          \[
            \tensor{(\nabla^{\gE} \bF)}{^{i}_{Jk}} \bydef \frac{\partial \tensor{F}{^{i}_{J}}}{\partial {x^{k}}}.
          \]
  \end{itemize}
\end{thm}

\begin{rem}\label{rem:pgradMMC}
  In Classical Continuum Mechanics, one starts with an embedding
  \[
    \pp \colon \body \to \lab,
    \qquad
    \bX \mapsto \xx=\pp(\bX),
  \]
  possibly parameterized by time, and introduces the \emph{deformation gradient}
  \[
    \bF_{\pp}(\bX) \bydef T_{\bX}\pp \colon T_{\bX}\body \to T_{\pp(\bX)}\lab.
  \]
  In our setting, if we assume that the restriction $\Psi_{t} \colon \Omega_{t} \to \body$ is a diffeomorphism and set $\pp \bydef {\Psi_{t}}^{-1}$, then we have $\bF_{\pp}(\bX) = \bF (\pp(\bX))$ and
  \[
    \frac{\partial \tensor{(\bF_{\pp})}{^{i}_{J}}}{\partial {X^{K}}} = \frac{\partial \tensor{F}{^{i}_{J}}}{\partial {x^{l}}}\tensor{(\bF_{\pp})}{^{l}_{K}},
  \]
  or, in tensorial notation,
  \[
    \nabla\bF_{\pp} = \nabla^{\gE} \bF \cdot \bF_{\pp} = \nabla^{\gE} \bF \cdot \bF \circ \pp.
  \]
  This means that the left inner inside term of \eqref{eq:HMink} can be rewritten
  \[
    \bF^{-1} \cdot \nabla^{\gE} \bF \cdot \bF =\left( \bF_{\pp}^{-1} \cdot \nabla \bF_{\pp} \right) \circ \pp^{-1},
  \]
  so as to display the standard deformation gradient variable $\bF_{\pp}^{-1} \cdot \nabla \bF_{\pp}$ on $\body$
  of Classical Continuum Mechanics.
\end{rem}

\begin{proof}
  We have first the orthogonal decompositions relative to $\bN$
  \[
    g^{-1} = \gE^{-1} - \bN \otimes \bN, \quad T\Psi = \frac{1}{c} (\bF^{-1}\uu) \otimes \bN^\flat + \bF^{-1}, \quad T\Psi^{\star} = \frac{1}{c} \bN^\flat \otimes (\bF^{-1}\uu) + \bF^{-\star}.
  \]
  We get thus
  \[
    \bK =  T\Psi \cdot g^{-1} \cdot T\Psi^\star =  \bF^{-1}\gE^{-1}\bF^{-\star} - \frac{1}{c^{2}} (\bF^{-1}\uu) \otimes (\bF^{-1}\uu) = \bF^{-1} \left(\gE^{-1} - \frac{1}{c^{2}}{\uu\otimes\uu}\right)\bF^{-\star}.
  \]

  We will now particularize the formulas given in Lemma~\ref{lem:hessian} when $g = \eta$. In that case, we have
  \[
    \nabla^{\eta} \bN = 0 \quad \text{and} \quad c\,\UU = \gamma\left(c\bN + \uu\right), \quad \text{with} \quad \bN = \frac{1}{c}\pd{t}{}.
  \]
  Note also that since $\nabla^\eta\bN=0$ and $T\Psi\cdot\UU=0$, we get that
  \begin{align*}
    \nabla^{\eta}_{\XX} (c\UU) & = \nabla^{\eta}_{\XX} \big(\gamma (c\bN + \uu)\big)
    \\
                               & = (\dd \gamma\cdot \XX) (c\bN + \uu) + \gamma \nabla^{\eta}_{\XX}(c\bN + \uu)
    \\
                               & = \frac{1}{\gamma}(\dd \gamma \cdot \XX) c\UU + \gamma \nabla^{\eta}_{\XX} \uu,
  \end{align*}
  for any vector field $\XX$. Hence
  \[
    T\Psi \cdot \nabla^{\eta} (c\UU) = \gamma \, T\Psi \cdot \nabla^{\eta} \uu.
  \]

  We therefore get, using Remark~\ref{rem:material-derivative},
  \[
    \bA = T\Psi\cdot\nabla^{g}_{c\UU}(c\UU) = \gamma \, T\Psi \cdot \nabla^{\eta}_{c\UU} \uu = \gamma^{2} \, T\Psi_{t} \cdot \dot{\uu} = \gamma^{2} \, \bF^{-1} \dot{\uu}=\bF^{-1} \ba,
  \]
  then, using \eqref{eq:UbF},
  \begin{align*}
    \bL & = T\Psi \cdot (\nabla^{g} c\UU) \cdot \bF + \frac{1}{c} \bA \otimes (\UU^\flat \bF)
    \\
        & = \gamma \, T\Psi \cdot \nabla^{\eta} \uu \cdot \bF + \frac{\gamma}{c^{2}} \, (\bF^{-1} \ba) \otimes (\uu^\flat \bF)
    \\
        & =  \gamma \, \bF^{-1} \cdot  \nabla^{\gE} \uu \cdot \bF + \frac{\gamma}{c^{2}} \, \bF^{-1} \cdot (\ba \otimes \uu^\flat) \cdot \bF
    \\
        & = \bF^{-1} \cdot \left(\gamma \, \nabla^{\gE} \uu + \frac{\gamma}{c^{2}} \, \ba \otimes \uu^\flat \right) \cdot \bF
  \end{align*}
  and, using \eqref{eq:UbF},
  \[
    \bH  = (T\Psi\cdot\nabla^{\eta}_{\FF_{J}}(\FF_{K})) \otimes \dd X^{J} \otimes \dd X^{K} + \frac{2\gamma }{c^{2}} \big(\bL \otimes (\uu^\flat\bF)\big)^{(23)} + \frac{\gamma^{2}}{c^{4}} \bA \otimes (\uu^\flat\bF) \otimes (\uu^\flat\bF).
  \]
  But
  \[
    T\Psi\cdot\nabla^{\eta}_{\FF_{J}}(\FF_{K}) = T\Psi_{t}\cdot\nabla^{\gE}_{\FF_{J}}(\FF_{K}) = \bF^{-1}\cdot\nabla^{\gE}_{\FF_{J}}(\FF_{K}),
  \]
  where we have defined $\FF_{J}$ as $\bF \cdot \partial/\partial X^{J}$ and where $\bF(m) = (T_{m}\Psi_{t})^{-1}$ for each $m \in \Omega_{t}$. We can thus write in components
  \[
    (\FF_{J})^{i} = \tensor{F}{^{i}_{J}}, \quad \text{and} \quad  (\nabla^{\gE}_{\FF_{J}}(\FF_{K}))^{i} = \frac{\partial (\FF_{K})^{i}}{\partial {x^{l}}} (\FF_{J})^{l} = \frac{\partial \tensor{F}{^{i}_{K}}}{\partial {x^{l}}} \tensor{F}{^{l}_{J}}.
  \]
  We get thus
  \[
    \left(\bF^{-1}\nabla^{\gE}_{\FF_{J}} \FF_{K}\right)^{I} = \left(\bF^{-1}\nabla^{\gE}_{\FF_{K}} \FF_{J}\right)^{I} = \tensor{(\bF^{-1})}{^{I}_{i}} \frac{\partial \tensor{F}{^{i}_{J}}}{\partial {x^{l}}} \tensor{F}{^{l}_{K}} = \tensor{(\bF^{-1})}{^{I}_{i}} \tensor{(\nabla^{\gE} \bF)}{^{i}_{Jl}} \tensor{F}{^{l}_{K}},
  \]
  if we set
  \[
    \tensor{(\nabla^{\gE} \bF)}{^{i}_{Jk}} \bydef \frac{\partial \tensor{F}{^{i}_{J}}}{\partial {x^{k}}}.
  \]
  We have finally
  \[
    \bH  = \bF^{-1} \cdot \nabla^{\gE} \bF \cdot \bF + \frac{2 \gamma}{c^{2}} \big(\bL \otimes (\uu^\flat\bF)\big)^{(23)} +\frac{\gamma^{2}}{c^{4}} \bA \otimes (\uu^\flat\bF) \otimes (\uu^\flat\bF).
    \qedhere
  \]
\end{proof}

\section{Classical limit of Gradient Relativistic Elasticity}
\label{sec:classical-limit}

\subsection{Classical limit of the relativistic invariants}

In this section, we consider pointwise convergence, as we have not defined any function space to which the fields under consideration must belong, even though we implicitly assume that these fields are smooth. As the speed of light $c$ goes to infinity, the conformation tends towards the inverse of the Cauchy--Green tensor $\bC\bydef\bF^{\star} \gE\, \bF$ of Classical Continuum Mechanics \cite{Sou1958,Mau1978b},
\[
  \lim_{c\to\infty} \bK= \bC^{-1}.
\]

The vector $\bA$ depends on the speed of light $c$ through the Lorentz factor $\gamma$.
In the classical limit, since $ \lim_{c\to \infty} \gamma~=~1$, we then have
\[
  \lim_{c\to \infty} \bA=\bF^{-1}\, \ba ,
  \qquad
  \ba=\dot \uu= \partial_{t} \uu + \nabla^{\gE}_{\uu} \uu,
\]
meaning that $\bA$, interpreted as the pullback to $\body$ of the Eulerian acceleration $\ba$, is the material acceleration of Classical Continuum Mechanics.

At the classical limit
\[
  \lim_{c\to \infty} \bL
  =
  \bF^{-1}  (\nabla^{\gE} \uu) \,\bF
  =\bF^{-1}\pmb{\ell}\,\bF
  ,\qquad
  \pmb{\ell}:=\nabla^{\gE} \uu,
\]
This tensor is interpreted as the pullback to $\body$ of the so-called Eulerian velocity gradient $\pmb{\ell}:=\nabla^{\gE}\uu$ (a mixed second-order tensor on $\Omega_{t}$). It is the material (on $\body$) velocity gradient of Classical Continuum Mechanics. The infinite speed of light limit of $\bL$ is thus built from the sum of the (symmetric) strain rate tensor $\bd:=\pmb{\ell}^{s}=(\nabla^{\gE}\uu)^{s}$, which is known to be objective, and of the (skew-symmetric) spin tensor $\bw:=\pmb{\ell}^{a}=(\nabla^{\gE}\uu)^{a}$, which is known to be non-objective.

Lastly, the contravariant third-order tensor $\bH$ has the following limit as the speed of light tends to infinity
\[
  \lim_{c\to \infty} \bH \circ \pp
  =
  \bF_{\pp}^{-1}\nabla \bF_{\pp},
  \qquad
  \bF_{\pp}=T\pp,
\]
where $\pp\colon \body \to \lab$ is the embedding of the body in the Euclidean space (see \autoref{rem:pgradMMC}). The secondary variable $\bF_{\pp}^{-1}\nabla \bF_{\pp}$, obtained as such a limiting case for an infinite speed of light, is the fundamental variable (objective, see the discussion below) of the Toupin--Mindlin finite strain second gradient theory~\cite{Eri1999}. It is equal to $\bF_{\phi}^{-1} \nabla \bF_{\phi}$ when the body $\body$ is identified with the reference configuration $\Omega_{0}$, $\phi\colon \Omega_{0}\to \Omega_{t}$ is the deformation and $\bF_{\phi}=T\phi$ is the deformation gradient.

\begin{rem}
  The alternative (contravariant) invariants $\bL^{\sharp_{\bK^{-1}}} = \bL\cdot \bK$ and $\bH^{\sharp_{\bK^{-1}}}$ introduced in \autoref{rem:alternative-invariants} have the following classical limits (at $c\to \infty$)
  \[
    \lim_{c\to \infty} \bL^{\sharp_{\bK^{-1}}} = \bF^{-1}  (\nabla^{\gE} \uu) \,\bF\cdot \bC^{-1}
    \quad \text{and}\quad
    \lim_{c\to \infty} \bH^{\sharp_{\bK^{-1}}} \circ \pp=
    \left(\bF_{\pp}^{-1}\nabla \bF_{\pp}\right)^{\sharp_{\bC^{-1}}},
  \]
  the latter of components
  \[
    {\left(\bF_{\pp}^{-1}\nabla \bF_{\pp}\right)^{I}}_{PQ} (\bC^{-1})^{PI} (\bC^{-1})^{QJ}.
  \]
\end{rem}

\subsection{Second Gradient Elasticity from Variational Relativity}

The computations carried out in~\cite{Sou1964,DGK2026}, related to the introduction of an observer and of a reference configuration $\Omega_{0}$ identified with the body $\body$, show that the Variational Relativity Lagrangian \eqref{eq:matterL}--\eqref{eq:matterL0} defined on the World tube $\mW$ ($\rho_{r}$ being the rest mass density, function of fixed $\rho_{\circ}$ by mass conservation \eqref{eq:rhor})
\[
  \mL^{\text{matter}}= \int_{\mW} L_{0} \vol_{g}
  , \qquad
  L_{0}=\rho_{r} \left(c^{2} + w\right),
\]
recasts in the Minkowski spacetime as
\[
  \mL^{\text{matter}}[\phi]=\int_{[t_{0}, t_{1}]}
  \left( mc^{2} +\int_{\Omega_{0}} (W-K)\vol_{\bq} \right) \textrm{d}t.
\]
Here,
$\phi=\pp=\Psi_{t}^{-1}\colon \Omega_{0} \to \Omega_{t}$ is the Continuum Mechanics transformation (by Remark~\ref{rem:pgradMMC}, and set $\bF:=T\phi$) and \cite[Eq.~(5.3)--(5.4)]{DGK2026}
\begin{align*}
   & K=
  \rho_{\circ} c^{2} \left(1- \Gamma^{-1} \right)= \frac{1}{2} \norm{\bm{V}}_{\bq}^{2} + O\left(\frac{1}{c^{2}}\right),
  \\
   & W= \rho_{\circ} \left(w \circ \phi\right)\Gamma^{-1} ,
  \\
   & \Gamma^{-1}=\sqrt{1- \frac{1}{c^{2}} \norm{\partial_{t} \phi}_{\bq}^{2}}
  =1 -\frac{1}{2 c^{2}} \norm{\bm{V}}_{\bq}^{2}+ O\left(\frac{1}{c^{4}}\right),
\end{align*}
are, respectively, the relativistic kinetic energy density ($\bm{V}=\partial_{t} \phi= \uu \circ \phi$ being the relativistic Lagrangian velocity), the relativistic strain energy density and the inverse of the Lorentz factor on the reference configuration.
Recognize $m=\int_{\Omega_{t}} \rho \vol_{\bq}=\int_{\Omega_{t}} \rho_{r} \gamma \vol_{\bq}$
as the mass of the studied continuous medium.

Considering the generalized Mooney density for relativistic second gradient elastic media (with $A$, $B$ and the internal length $l_{\text{int}}$ as material parameters)
\[
  w=w\left(\bK^{-1}, \bH\right)= A \left(\tr \left(\bq^{-1} \bK^{-1}\right)-3\right) + B \, l_{\text{int}} \, \norm{\bH}_{\bq}^{2},
\]
gives, by Theorem~\ref{thm:limits},
\begin{align*}
  W=\, & \rho_{\circ} A \,  \left(\tr \left[ \bF^{-1} \left(\bq^{-1}-\frac{1}{c^{2}}\bm{V}\otimes\bm{V}\right)\bF^{-\star}\bq \right]^{-1}-3\right)
  \Gamma^{-1}
  \\
       & + \rho_{\circ} B \, l_{\text{int}} \, \left\|
  \bF^{-1} \nabla \bF  + \frac{2\Gamma}{c^{2}} \big(\bL\circ \phi \otimes (\bm{V}^\flat\bF)\big)^{(23)} + \frac{\Gamma^{2}}{c^{4}} \bA \circ \phi \otimes (\bm{V}^\flat\bF) \otimes (\bm{V}^\flat\bF)
  \right\|_{\bq}^{2} \Gamma^{-1}
\end{align*}
where
\[
  \bL=\bF^{-1}  (\nabla^{\gE} \uu) \,\bF + O\left(\frac{1}{c^{2}}\right),
  \qquad
  \bA= \bF^{-1} \dot \uu
  + O\left(\frac{1}{c^{2}}\right).
\]
In the limit of infinite speed of light, we get from the above Galilean limits
\[
  \lim_{c\to \infty} W= \rho_{\circ} A \left(\tr \left(\bq^{-1} \bC\right)-3\right) + \rho_{\circ}  B\, l_{\text{int}} \, \norm{\bF^{-1} \nabla \bF}_{\bq}^{2},
\]
which is a particular case of the energy density of Classical Continuum Mechanics used in~\cite{CCT2004}.

\begin{rem}
  The $1/c^{2}$ terms in the relativistic strain energy density $W$ are functions of the non-objective limits (by Theorem~\ref{thm:limits}) of the Lagrangian velocity $\bm{V}$, of the material acceleration $\bA$ and of the velocity gradient~$\bL$.
\end{rem}

\subsection{From General Covariance to Objectivity}

Starting from General Relativity and fixing the metric to the flat Minkowski metric $\eta$ on $\univ \simeq \RR^{4}$ corresponds to the framework of Special Relativity.
This breaks the symmetry, because in this setting only diffeomorphisms $\varphi$ which fix $\eta$, meaning that $\pull\eta = \eta$, are allowed. These diffeomorphisms are the isometries of the Minkowski space and constitute the \textsl{Poincaré group}, a finite-dimensional Lie group of dimension $10$.
In Classical Continuum Mechanics, the Universe $\univ \simeq \RR^{4}$ is no longer a Lorentzian manifold. It is equipped with a \textsl{Galilean structure}, \textit{i.e.}, a pair $(\bkappa,\theta)$, where $\bkappa$ is a symmetric second-order contravariant tensor of signature $(0,+,+,+)$ (the classical spatial cometric) and $\theta$ is a $1$-form which spans the kernel of $\bkappa$ (the clock) \cite{KD2023}. In the present case, this Galilean structure $(\bkappa,\theta)$ is obtained as the limit of the $1$-parameter family of smooth Lorentz metrics
\[
  \eta_{c}^{-1} = -c^{-2}{\partial_{t}}^{2} + {\partial_{x}}^{2} + {\partial_{y}}^{2} + {\partial_{z}}^{2},
\]
when $c \to \infty$, and where $\theta = \dd t$.
Diffeomorphisms $\varphi$ which preserve the Galilean structure
\[
  \pull\bkappa = \bkappa, \quad \text{and} \quad \pull\theta = \theta,
\]
are called \textsl{Galileomorphisms}. In the canonical coordinate system of $\RR^{4}$, they are given by
\[
  \overline{\xx} = \bQ(t)\xx + \vb(t), \qquad \overline{t} = t + c,
\]
where $\xx = (x,y,z)$, $\bQ(t)$ is a rotation and $\vb(t)$ is a vector, both depending on time. These transformations are the ``change of observer'' of Classical Mechanics. Those quantities which are invariant under these transformations are called \textsl{objective} or \textsl{frame independent} \cite{KD2024}.
The four-dimensional derivation in \cite{dIMS2009} of (non-relativistic) boundary conditions that are Galilean invariants can be seen as a classical counterpart of the question raised in the present work about the objectivity of the Galilean limits.

As a result of our study, it appears that even if all the quantities of \autoref{cor:second_gradient} are invariant under general (orientation-preserving) diffeomorphisms, their limit may or may not be objective. The reason is that some quantities, like the conformation $\bK$ or the material relativistic strain gradient $\bH$, defined \textit{a priori} for Lorentzian metrics $g$, extend continuously to the boundary of this set (that is, to degenerate Lorentzian cometrics) and thus to Galilean structures.
In that case, the General Covariance of $\bK$ and $\bH$ leads to the objectivity of their limits $\bC^{-1}$ and $\bF_{\pp}^{-1}\nabla \bF_{\pp}$, by restricting to symmetries of the Galilean structure. However, this is not the case for other invariants like the material relativistic acceleration $\bA$ and the material relativistic velocity gradient $\bL$.

Such results in fact clarify the criticism levelled at the principle of objectivity, which is sometimes considered too restrictive, for instance for viscous fluids
\cite{Ede1973,Ast1979,Mur1983,Spe1998}: indeed, the generally covariant secondary variables which do not have objective limits are in fact a first time gradient and a first spatial gradient (for $\bA$ and $\bL$). All generally covariant first-order ($\bK$) and second-order $(\bH)$ spatial gradient variables allowed by \autoref{thm:second_gradient} for second-gradient elasticity of solids have objective limits at $c\to\infty$.
The Variational Relativity framework does not impose any more constraints on admissible second-order constitutive laws than the standard application of the principle of objectivity of Classical Continuum Mechanics.

\section{Conclusion}

We have extended Souriau's theory of Relativistic first gradient elasticity to a second gradient theory. We have shown that, assuming General Covariance for the Lagrangian density $L$, it is in fact a function of the following diffeomorphism invariants: the matter field $\Psi$ itself, Souriau's conformation $\bK$ and additional three-dimensional tensor-valued invariants, namely
\begin{itemize}
  \item the material relativistic acceleration (vector $\bA$, Eq. \eqref{eq:AHess}),
  \item the material relativistic velocity gradient (second-order tensor $\bL$, \eqref{eq:LHess}),
  \item the material relativistic strain gradient (third-order tensor $\bH$, Eq. \eqref{eq:HHess}),
  \item relativistic gravitation/matter coupling tensors (second, third and fourth-order tensors $\bM_{p}$, Eq. \eqref{eq:M4}--\eqref{eq:M2}), built from the curvature tensor $\bR^{g}$.
\end{itemize}
They play the role of Souriau's conformation, but for a relativistic second gradient theory. The Galilean limits of these invariants have been computed in the flat Minkowski spacetime of Special Relativity. Our results show that the 3D Classical Continuum Mechanics second gradient theory can be derived as the limiting case $c\to \infty$ of such a relativistic theory.
They show that some secondary variables converge in the Galilean limit to objective quantities (such as the fundamental variable $\bH\to \bF_{\phi}^{-1}\nabla \bF_{\phi}$ of finite strain second gradient theory, modeling solids), while other secondary variables converge to non-objective quantities (modeling fluids). The present work thus contributes to the debate on the criticism of the Principle of Objectivity, sometimes considered too restrictive, particularly in fluid mechanics \cite{Ede1973,Ast1979,Mur1983,Spe1998}.

The novel gravitation/matter coupling tensors $\bM_{p}$ are zero in the flat Minkowski spacetime, thus at the Galilean limit. It is possible, but non-trivial, to take the more general Newton--Cartan limit at infinite speed of light in a curved static spacetime (in which the gravitation/matter coupling tensors $\bM_{p}$ and their limit at $c\to \infty$ will not vanish). The procedure for achieving this is detailed in Ref.~\cite{Kuen1976} and in the recent work \cite{KD2023} for the curved Schwarzschild spacetime.

At the astrophysical scale, the generally covariant definitions \eqref{eq:M4}--\eqref{eq:M2} for the curvature-based arguments $\bM_{p}$ of the matter Lagrangian density may help model the coupling between gravitation and dark energy/matter. As an alternative, the cosmological gradient fluids (with the gradient terms of Section \ref{sec:gradient_fluids}) are natural candidates as models of dark energy. Corollary \ref{cor:second_gradient} may also prove useful in the study of gravitational wave damping related to matter, or of neutron star oscillations, introducing---through the gradient term $\bH$---an internal length in their mechanical behavior.

We finally leave to future work the question of which junction conditions should be prescribed at the boundary between the four-dimensional World tube (interior solution) and the vacuum (exterior solution), or possibly another World tube, and the corresponding boundary conditions in the Galilean limit. The spacetime boundary conditions cannot be guessed but must be deduced from the variational principle. On that topic, we refer in General Relativity to \cite{GH1977,Hay1993} and \cite[Chap. 8]{Gou2012} for prescribed spacetime conditions on the boundary of a World tube, and in Classical Continuum Mechanics to \cite{dIMS2009} for the variational deduction of the boundary conditions at a fluid-permeable interface through which matter passes.

\appendix

\section{Covariant derivatives, Torsion, Curvature and Hessian tensors}
\label{sec:nabla-and-tensors}

\subsection{Covariant derivatives}
A \textsl{covariant derivative} $\nabla$ on $\univ$ is an operator which extends the usual derivative of scalar fields
\[
  \nabla_{\XX} f
  \bydef
  \dd f\cdot \XX,
  \qquad
  f \in \Cinf(\univ),
  \quad
  \XX  \in \mathfrak{X}^{1}(\univ),
\]
to tensor fields. It is characterized by the following two properties:
\[
  \nabla_{f\XX} \YY = f \nabla_{\XX} \YY,
  \qquad
  \nabla_{\XX} (f\YY) = \nabla_\XX(f) \YY + f\nabla_{\XX}(\YY),
\]
for all functions $f \in \Cinf(\univ)$ and all vector fields $\XX$, $\YY$ on $\univ$. Using the Leibniz rule, it extends naturally to covector fields $\alpha \in \Omega^{1}(\univ)$, via the formula
\[
  (\nabla_{\XX}\alpha) (\YY)
  \bydef
  \nabla_{\XX} (\alpha(\YY)) - \alpha(\nabla_{\XX} \YY),
\]
and more generally to all tensor fields. We admit that if $\bt$ is a tensor field of order $r$, then $\nabla\bt$ is also a tensor field, of order $(r+1)$.
\begin{exam}
  For a metric $g$, Leibniz rule yields
  \[
    (\nabla_{\pmb{Z}} g)(\XX,\YY)
    \bydef
    \dd(g(\XX,\YY))\cdot\pmb{Z} - g(\nabla_{\pmb{Z}} \XX,\YY) - g(\XX,\nabla_{\pmb{Z}} \YY).
  \]
\end{exam}

Given a local coordinate system $(x^{\mu})$, a covariant derivative $\nabla$ is determined by its \textsl{Christoffel coefficients} (of the second kind)
\[
  \tensor{\Gamma}{^{\kappa}_{\lambda\mu}}
  \bydef
  \dd x^{\kappa}
  \left(\nabla_{\partial_{\mu}}\partial_{\lambda} \right)
\]
Note that despite the tensorial notation, these are not the components of a tensor.
\begin{defn}
  Denote by $[\XX, \YY]\in\mathfrak{X}^1(\univ)$ the Lie bracket of two vector fields $\XX$ and $\YY$ on $\univ$.
  Given a covariant derivative $\nabla$, one defines the following tensor fields.
  \begin{enumerate}
    \item the \textsl{torsion tensor}
          \[
            \bT^\nabla\left(\XX,\YY\right)
            \bydef
            \nabla_\XX \YY - \nabla_\YY \XX - [\XX,\YY],
          \]
          with components
          \[
            \tensor{T}{^{\kappa}_{\lambda\mu}}
            =
            \tensor{\Gamma}{^{\kappa}_{\lambda\mu}} - \tensor{\Gamma}{^{\kappa}_{\mu\lambda}}\ ;
          \]
    \item the \textsl{Riemann curvature tensor}
          \[
            \bR^\nabla\left(\XX,\YY\right) \pmb{Z}
            \bydef
            \nabla_\XX \nabla_\YY \pmb{Z} - \nabla_\YY \nabla_\XX \pmb{Z} - \nabla_{[\XX,\YY]} \pmb{Z} ,
          \]
          with components
          \[
            \tensor{R}{^{\kappa}_{\lambda\mu\nu}}
            =
            \partial_{\mu}\tensor{\Gamma}{^{\kappa}_{\nu\lambda}}
            -
            \partial_{\nu}\tensor{\Gamma}{^{\kappa}_{\mu\lambda}}
            +
            \tensor{\Gamma}{^{\kappa}_{\mu\rho}}\tensor{\Gamma}{^{\rho}_{\nu\lambda}}
            -
            \tensor{\Gamma}{^{\kappa}_{\nu\rho}}\tensor{\Gamma}{^{\rho}_{\mu\lambda}}\ ;
          \]
    \item the \textsl{Hessian tensor} of a scalar field $f$
          \[
            \Hess^\nabla\! f
            \bydef
            \nabla\nabla f
            =
            \nabla(\dd f),
          \]
          with components
          \[
            \tensor{(\Hess^{\nabla}\!f)}{_{\mu\nu}}
            =
            \partial_\mu\partial_\nu f - \partial_\kappa f\, \tensor{\Gamma}{^{\kappa}_{\nu\mu}}.
          \]
  \end{enumerate}
\end{defn}

A covariant derivative $\nabla$ is said to be \textsl{symmetric} if its torsion tensor $\bT^\nabla$ vanishes or, in other words, if $\tensor{\Gamma}{^{\kappa}_{\lambda\mu}}$ is symmetric in the lower indices. In that case, the Hessian $\Hess^\nabla\! f$ of any scalar field is also symmetric.

Recall that the Universe $(\univ,g)$ is assumed to be a pseudo-Riemannian manifold where $g$ is a Lorentzian metric of signature $(-+++)$. Such a manifold is endowed with a natural covariant derivative $\nabla^{g}$, the \textsl{Levi-Civita connection}, which is uniquely defined by the following two properties:
\[
  \nabla^{g} g = 0, \quad \text{and} \quad \bT^{g} = 0,
\]
where $\bT^{g}\bydef\bT^{\nabla^g}$ is the torsion of $\nabla^{g}$.
The Levi-Civita connection is determined by the equations
\[
  \tensor{\Gamma}{^{\kappa}_{\lambda\mu}}
  =
  \frac{\tensor{g}{^{\kappa\rho}}}{2}
  \left( \tensor{g}{_{\rho\lambda,\mu}} + \tensor{g}{_{\rho\mu,\lambda}} - \tensor{g}{_{\lambda\mu,\rho}} \right),
\]
that can be inverted as follows:
\[
  \tensor{g}{_{\kappa\lambda,\mu}} = \tensor{g}{_{\kappa\rho}} \tensor{\Gamma}{^{\rho}_{\lambda\mu}} + \tensor{g}{_{\lambda\sigma}} \tensor{\Gamma}{^{\sigma}_{\kappa\mu}}.
\]
One concludes therefore that, $g$ being invertible, the data of $\set{g_{\mu\nu}, \partial_\mu\tensor{g}{_{\kappa\lambda}},\partial_\mu\partial_\nu\tensor{g}{_{\kappa\lambda}}}$ is equivalent to that of
$\set{g_{\mu\nu}, \tensor{\Gamma}{^{\kappa}_{\lambda\mu}},\partial_\nu\tensor{\Gamma}{^{\kappa}_{\lambda\mu}}}$.

\subsection{Action of the diffeomorphism group}
\label{subsec:diff-action-tensors}

Following Felix Klein's Erlangen Program~\cite{Kle1974}, a geometry is defined by a group action. As such, Differential Geometry is defined by the action of the diffeomorphism group of a manifold. One of the main questions dealt with in this work is covariance under the diffeomorphism group~\cite{KSM1993}. Let us describe how tensors and connections transform under the action of this group.

Let $\varphi\colon U\to V$ be a smooth map. The differential $T\varphi$ of $\varphi$ allows one to \textsl{push-forward} a contravariant vector $\XX$ at $m\in U$ to a contravariant vector $T_{m}\varphi \cdot \XX$ at $\varphi(m)$. Globalizing this construction for a vector field $\XX$ on $U$, one obtains a section $T\varphi \cdot \XX$ of the pullback bundle
\[
  \pull TV \bydef \coprod_{m \in U} T_{\varphi(m)}V.
\]
When $\varphi$ is a diffeomorphism, for a vector field $\XX$ on $U$, one can define a vector field $\push \XX$ on $V$ by setting at a point $p\in V$:
\[
  (\push \XX)_p
  \bydef
  T_{\varphi^{-1}(p)}\varphi \cdot \XX_{\varphi^{-1}(p)}.
\]
One can extend this construction diagonally to higher-order contravariant tensors to obtain a linear map $\push $ from contravariant tensor fields on $U$ to contravariant tensor fields on $V$.

Let now $\omega$ be a covariant tensor field on $V$ (or a scalar field). For a smooth map $\varphi\colon U\to V$, there is a natural covariant tensor field of the same type on $U$, obtained by pullback by $\varphi$. It is given at a point $m\in U$ by:
\[
  (\pull\alpha)_{m}(\XX_1,\dotsc,\XX_k)
  \bydef
  \alpha_{\varphi(m)}(T_m\varphi \cdot \XX_1,\dotsc,T_m\varphi\cdot  \XX_k).
\]
For $1$-forms, using $\star$ for the dual transpose, this action can also be written as
\[
  \pull\alpha = (T\varphi)^\star \cdot \alpha \circ \varphi.
\]
For exact $1$-forms, one gets
\[
  (\pull\dd f)_{m}(\XX)
  =
  \dd f_{\varphi(m)}(T_m\varphi \cdot \XX)
  =
  \dd(f\circ\varphi)_{m} \cdot \XX
  =
  \dd(\pull f)_{m}\cdot \XX,
\]
which we sum up as $\pull\dd=\dd\pull$ (exterior derivatives and pullbacks commute).

When $\varphi$ is a diffeomorphism, one can also define the pullback to $U$ of a contravariant tensor field $\XX$ on $V$ as the pushforward by the diffeomorphism $\varphi^{-1}$. For instance for a vector field:
\[
  (\pull \XX)_m
  \bydef
  ((\varphi^{-1})_*\XX)_m
  =
  T_{\varphi(m)}\varphi^{-1} \cdot \XX_{\varphi(m)}.
\]
Along the same lines, one can define the pushforward of a covariant tensor field $\omega$, by pulling back by the inverse diffeomorphism as in $\push\omega=(\varphi^{-1})^*\omega$.

The right action $\pull$ of the diffeomorphism group (the \textsl{pullback}) and its inverse left action $\push $ (the \textsl{pushforward}) can be extended diagonally to arbitrary mixed tensor fields by using the tensor product decomposition of the tensor bundle. As intended, the pullback and the pushforward commute with contractions.

\begin{exam}
  The pullback of a metric $g\colon \mathfrak{X}^1(\univ)\to\Omega^{1}(\univ)$ is given by
  \[
    (\pull g)_{m}
    =
    (T\varphi)^{\star} \cdot g_{\varphi(m)} \cdot T\varphi,
  \]
  and the pullback of a cometric $g^{-1}\colon\Omega^{1}(\univ)\to\mathfrak{X}^1(\univ)$ is given by
  \[
    (\pull g^{-1})_{m}
    \bydef
    (T\varphi)^{-1} \cdot (g^{-1})_{\varphi(m)} \cdot (T\varphi)^{-\star}.
  \]
  One infers for instance that the \textsl{gradient} vector field $\grad^g f \bydef g^{-1}\dd f$ is generally covariant of the pair of variables $(g,f)$. This means that:
  \begin{align*}
    \grad^{\varphi^{*}g} \varphi^{*}f & = (\pull g)^{-1} \dd(\pull f)
    \\
                                      & =  T\varphi^{-1} \cdot g^{-1} \cdot T\varphi^{-\star} \cdot T\varphi^{\star} \cdot \dd f
    \\
                                      & = (T\varphi)^{-1} \cdot g^{-1}\dd f
    \\
                                      & = \pull(g^{-1}\dd f)
    \\
                                      & = \pull (\grad^g f).
  \end{align*}
\end{exam}

As already mentioned, covariant derivatives are not tensor fields. However, it is also possible to define pullbacks of covariant derivatives. To be natural, this definition is required to satisfy
\[
  \pull\left(\nabla_\XX \YY \right)
  =
  \left(\pull\nabla \right)_{\pull \XX} \left(\pull \YY \right).
\]
which yields the following definition of the pullback of a covariant derivative:
\begin{equation}
  \label{eq:pullback_covariant_derivative}
  \left(\pull\nabla\right)_{\XX} \YY
  \bydef
  \pull\left({\nabla}_{\push \XX}\push \YY\right).
\end{equation}

\begin{lem}
  The pullback of a covariant derivative is still a covariant derivative. Moreover, for any tensor field $\bt$, one has
  \[
    \pull\left(\nabla_\XX \bt \right)
    =
    \left(\pull\nabla \right)_{\pull \XX} \left(\pull \bt \right).
  \]
\end{lem}

\begin{proof}
  Recall that a covariant derivative coincides with $\dd$ on scalar fields, and that once it is defined on contravariant vector fields, it extends in a unique way to the tensor bundle using the Leibniz rule. Let us thus check that \eqref{eq:pullback_covariant_derivative} indeed defines a covariant derivative on $T\univ$.
  One has
  \[
    (\pull\nabla)_{f\XX}(\YY)
    =
    \pull({\nabla}_{\push (f\XX)}(\push \YY))
    =
    \pull(\push f)\pull({\nabla}_{\push \XX}(\push \YY))
    =
    f
      (\pull\nabla)_{\XX}(\YY),
  \]
  and
  \[
    (\pull\nabla)_{\XX}(f\YY)
    =
    \pull({\nabla}_{\push \XX}\push (f\YY))
    =
    \pull({\nabla}_{\push \XX}(\push f\push \YY))
    =
    \pull(\dd\push f(\XX)\push \YY+\push f{\nabla}_{\push \XX}(\push \YY)).
  \]
  Recall that $\pull\dd=\dd\pull$, hence:
  \[
    (\pull\nabla)_{\XX}(f\YY)
    =
    \dd f(\XX)\YY+f(\pull\nabla)_{\XX}(\YY).
  \]

  Similarly the action of diffeomorphisms is the composition on scalar fields, and once it is defined on contravariant vector fields, it extends in a unique way to the tensor bundle using the diagonal action on contravariant tensor fields and the compatibility with contraction. We have already checked that for scalar fields $\pull\dd=\dd\pull$ and \eqref{eq:pullback_covariant_derivative} is an ad hoc definition to ensure that $\pull\nabla=(\pull\nabla)\pull$ on contravariant vector fields. The extension of this formula to all tensor fields is then formal (see \autoref{exam:nabla_pull} below).
\end{proof}

\begin{exam}\label{exam:nabla_pull}
  For a metric $g$, one has:
  \begin{align*}
    (\pull\nabla)(\pull g) & (\pull \XX,\pull \YY)
    \\
                           & =
    \dd\pull\left(g(\XX,\YY)\right)
    -
    (\pull g)((\pull\nabla)\pull \XX,\pull \YY)
    -
    (\pull g)(\pull \XX,(\pull\nabla)\pull \YY)
    \\
                           & =
    \pull\dd\left(g(\XX,\YY)\right)
    -
    (\pull g)(\pull(\nabla \XX),\pull \YY)
    -
    (\pull g)(\pull \XX,\pull(\nabla \YY))
    \\
                           & =
    \pull
    \left(
    \dd\left(g(\XX,\YY)\right)
    -
    g(\nabla \XX, \YY)
    -
    g(\XX,\nabla \YY)
    \right)
    \\
                           & =
    \pull((\nabla g)(\XX,\YY))
    =
    \pull(\nabla g)(\pull \XX,\pull \YY).
  \end{align*}
\end{exam}

We list below some formal consequences of \eqref{eq:pullback_covariant_derivative}, without proof.

\begin{cor}\label{cor:covariance_tensors}
  Let $\nabla$ be a covariant derivative. The following properties hold, for any local diffeomorphism $\varphi$ of $\univ$.
  \begin{enumerate}
    \item General Covariance of the torsion tensor:
          \[
            \bT^{\pull\nabla} = \pull\left(\bT^{\nabla}\right).
          \]
    \item General Covariance of the curvature tensor:
          \[
            \bR^{\pull\nabla} = \pull\left(\bR^{\nabla}\right).
          \]
    \item General Covariance of the Hessian tensor:
          \[
            \Hess^{\pull\nabla} \pull f = \pull(\Hess^{\nabla} f).
          \]
  \end{enumerate}
\end{cor}

Recall that on a pseudo-Riemannian manifold with metric $g$, we denote the Levi-Civita connection by $\nabla^{g}$.
The connection $\pull\nabla^{g}$ satisfies
$
  \bT^{\pull\nabla^{g}} = \pull\left(\bT^{\nabla^{g}}\right) = 0
$,
and
$
  \left(\pull\nabla^{g} \right) {\pull g}  = \pull\left( \nabla^{g} g\right) = 0
$.
Therefore, by the unique characterization of the Levi-Civita connection, we have
\[
  \nabla^{\pull g} = \pull\nabla^{g}.
\]
We can hence specialize the above results as follows:
\[
  \bR^{\pull g} = \pull\bR^{g},
  \quad \text{and} \quad
  \Hess^{\pull g} \pull f = \pull(\Hess^{g} f).
\]
It follows that the \textsl{Ricci tensor} $\Ric_{g} \bydef (R_{\lambda\nu})$ where $R_{\lambda\nu}\bydef\tensor{R}{^\kappa_{\lambda\kappa\nu}}$ and the \textsl{scalar curvature} $R_{g} \bydef g^{\lambda\nu}R_{\lambda\nu}$, obtained by contraction, also satisfy
\[
  \Ric_{\pull g} = \pull\Ric_{g}, \quad \text{and} \quad R_{\pull g} = \pull R_{g}.
\]

\subsection{Normalization of the Christoffel symbols}
\label{subsec:diff-action-jets}

We consider the Christoffel symbols of a connection $\nabla$ in local coordinates
\[
  \tensor{\Gamma}{^{\kappa}_{\lambda\mu}}
  =
  \dd x^{\kappa}
  \left(\nabla_{\partial_{\mu}}\partial_{\lambda} \right).
\]
and their jets.
In this section, for a local diffeomorphism $\varphi$, we adopt the notation
\[
  \tensor{\bar\Gamma}{^{\kappa}_{\lambda\mu}}
  \bydef
  \dd x^{\kappa}
  \left((\pull\nabla)_{\partial_{\mu}}\partial_{\lambda} \right),
\]
and we shall establish the following refinement of a classical normalization result (see also~\cite[(35.8)]{Sou1964} and~\cite[Section 2]{KS2021}).
\begin{thm}[Normalization of the Christoffel symbols]
  For any connection $\nabla$, and for any point $m\in\univ$, there exists a diffeomorphism with trivial linear part such that, at the point $m$,
  \[
    \tensor{\bar\Gamma}{^\kappa_{(\lambda\mu)}} (m) = 0
    \quad\text{and}\quad
    \tensor{\bar\Gamma}{^\kappa_{(\lambda\mu,\nu)}}(m) = 0,
  \]
  where the comma in the index is the traditional notation for the partial derivative, and where parentheses are shorthand for symmetrization.
\end{thm}
\begin{proof}
  Let us first set the context of the proof.
  We consider a local diffeomorphism $\varphi$ fixing a local chart around $m$, equipped with a local coordinate system $(x^\mu)$.
  At $\varphi(m)$, we define the Taylor coefficients
  \[
    {A^\kappa}_\lambda
    =
    \push(\partial_{\lambda}\varphi^\kappa),
    \quad
    {B^\kappa}_{\lambda \mu}
    =
    \push(\partial_{\mu}\partial_{\lambda}\varphi^\kappa),
    \quad
    {C^\kappa}_{\lambda \mu \nu}
    =
    \push(\partial_\nu\partial_{\mu}\partial_{\lambda}\varphi^\kappa).
  \]
  Note that the matrix $({A^\kappa}_\lambda)$ is by assumption invertible, and that by virtue of Schwarz's lemma, both ${B^\kappa}_{\lambda \mu}$ and ${C^\kappa}_{\lambda \mu \nu}$ are totally symmetric with respect to their lower indices.
  It follows from the very definitions of $\mathbf{A},\mathbf{B},\mathbf{C}$ that
  \[
    \left\{
    \begin{array}{lclclclcl}
      \nabla_{\push\partial_\mu}
      \tensor{A}{^\kappa_\lambda}
       & = &
      \nabla_{\push\partial_\mu}
      \push(\partial_{\lambda}\varphi^\kappa)
       & = &
      \push ( \nabla_{\partial_\mu} (\partial_{\lambda}\varphi^\kappa))
       & = &
      \push (\partial_\mu \partial_{\lambda}\varphi^\kappa)
       & = &
      \tensor{B}{^\kappa_{\lambda\mu}}
      \\
      \nabla_{\push\partial_\nu}
      \tensor{B}{^\kappa_{\lambda\mu}}
       & = &
      \nabla_{\push\partial_\nu}
      \push (\partial_\mu \partial_{\lambda}\varphi^\kappa)
       & = &
      \push( \nabla_{\partial_\nu} \partial_\mu \partial_{\lambda}\varphi^\kappa)
       & = &
      \push(\partial_\nu\partial_{\mu}\partial_{\lambda}\varphi^\kappa)
       & = &
      {C^\kappa}_{\lambda \mu \nu}
    \end{array}
    \right..
  \]

  We can now compute $\bar\Gamma$. By definition:
  \begin{equation}\label{eq:A}
    (\push\partial_\lambda)_{\varphi(m)}
    =
    (\partial_\lambda\varphi^\kappa)_{m}(\partial_\kappa)_{\varphi(m)}
    =
    (\tensor{A}{^\kappa_\lambda}\partial_\kappa)_{\varphi(m)}.
  \end{equation}

  By Leibniz rule, at a point $\varphi(m)$, one has
  \begin{align*}
    \nabla_{\push\partial_\mu}(\push\partial_\lambda)
     & =
    \nabla_{\push\partial_\mu}
    \left(\tensor{A}{^\sigma_\lambda}\partial_\sigma\right)
    =
    \nabla_{\push\partial_\mu} \left(\partial_\sigma\right) \tensor{A}{^\sigma_\lambda}
    +
    \nabla_{\push\partial_\mu}(\tensor{A}{^\sigma_\lambda}) \partial_\sigma
    \\
     & =
    \nabla_{\partial_\tau} \left(\partial_\sigma\right) \tensor{A}{^\tau_\mu}\tensor{A}{^\sigma_\lambda}
    +
    \tensor{B}{^\sigma_{\lambda\mu}} \partial_\sigma
    =
    \left(
    \tensor{\Gamma}{^{\rho}_{\sigma\tau}}\
    \tensor{A}{^\tau_\mu}
    \tensor{A}{^\sigma_\lambda}
    +
    \tensor{B}{^\rho_{\lambda\mu}}
    \right)
    \partial_\rho.
  \end{align*}

  Therefore, inverting \eqref{eq:A} to get
  \[
    \partial_\rho
    =
    \tensor{(\mathbf{A}^{-1})}{^\kappa_\rho}
    \push\partial_\kappa,
  \]
  and pulling back at $m$:
  \[
    (\pull\nabla)_{\partial_{\mu}}(\partial_\lambda)
    =
    \pull\left(\nabla_{\push\partial_\mu}(\push\partial_\lambda)\right)
    =
    \pull
    \left(
    \tensor{(\mathbf{A}^{-1})}{^\kappa_\rho}
    \left(
    \tensor{\Gamma}{^{\rho}_{\sigma\tau}}
    \tensor{A}{^{\tau}_\mu}
    \tensor{A}{^{\sigma}_{\lambda}}
    +
    \tensor{B}{^\rho_{\lambda\mu}}
    \right)
    \right)
    \partial_{\kappa}.
  \]
  We get the following expression for $\bar\Gamma$,
  \begin{equation}\label{eq:bar_gamma}
    \tensor{\bar\Gamma}{^\kappa_{\lambda\mu}}(m)
    =
    \tensor{(\mathbf{A}^{-1})}{^\kappa_\rho}
    \left(
    \tensor{\Gamma}{^{\rho}_{\sigma\tau}}
    \tensor{A}{^{\tau}_\mu}
    \tensor{A}{^{\sigma}_{\lambda}}
    +
    \tensor{B}{^\rho_{\lambda\mu}}
    \right)
    (\varphi(m)).
  \end{equation}
  Differentiating this functional equality, one first gets
  \[
    \tensor{\bar\Gamma}{^\kappa_{\lambda\mu,\nu}}(m)
    =
    (\push\partial_\nu)
    \left(
    \tensor{(\mathbf{A}^{-1})}{^\kappa_\rho}
    \left(
    \tensor{\Gamma}{^{\rho}_{\sigma\tau}}
    \tensor{A}{^{\tau}_\mu}
    \tensor{A}{^{\sigma}_{\lambda}}
    +
    \tensor{B}{^\rho_{\lambda\mu}}
    \right)
    \right)
    (\varphi(m)),
  \]
  from which one infers, using the Leibniz rule and the definitions of $\mathbf{A},\mathbf{B},\mathbf{C}$:
  \begin{multline*}
    \tensor{\bar\Gamma}{^\kappa_{\lambda\mu,\nu}}(m)
    =
    \tensor{(\mathbf{A}^{-1})}{^\kappa_{\alpha}}
    \tensor{B}{^\alpha_{\beta\nu}}
    \tensor{(\mathbf{A}^{-1})}{^\beta_{\rho}}
    \left(
    \tensor{\Gamma}{^{\rho}_{\sigma\tau}}
    \tensor{A}{^{\tau}_\mu}
    \tensor{A}{^{\sigma}_{\lambda}}
    +
    \tensor{B}{^\rho_{\lambda\mu}}
    \right)
    (\varphi(m))
    \\
    +
    \tensor{(\mathbf{A}^{-1})}{^\kappa_\rho}
    \left(
    \tensor{\Gamma}{^{\rho}_{\sigma\tau,\upsilon}}
    \tensor{A}{^{\upsilon}_{\nu}}
    \tensor{A}{^{\tau}_\mu}
    \tensor{A}{^{\sigma}_{\lambda}}
    +
    \tensor{\Gamma}{^{\rho}_{\sigma\tau}}
    \tensor{A}{^{\tau}_\mu}
    \tensor{B}{^{\sigma}_{\lambda\nu}}
    +
    \tensor{\Gamma}{^{\rho}_{\sigma\tau}}
    \tensor{B}{^{\tau}_{\mu\nu}}
    \tensor{A}{^{\sigma}_{\lambda}}
    +
    \tensor{C}{^\rho_{\lambda\mu\nu}}
    \right)
    (\varphi(m)).
  \end{multline*}

  We sum up the results obtained so far in the special case where $\varphi(m)=m$ and $T\varphi=\id$. At the point $m$, one has:
  \begin{equation}\label{eq:diff_action_gamma}
    \left\{
    \begin{aligned}
      \tensor{\bar\Gamma}{^\kappa_{\lambda\mu}}
       & =
      \tensor{\Gamma}{^{\kappa}_{\lambda\mu}}
      +
      \tensor{B}{^\kappa_{\lambda\mu}}
      \\
      \tensor{\bar\Gamma}{^\kappa_{\lambda\mu,\nu}}
       & =
      \tensor{\Gamma}{^{\kappa}_{\lambda\mu,\nu}}
      +
      \tensor{B}{^\kappa_{\rho\nu}} ( \tensor{\Gamma}{^{\rho}_{\lambda\mu}} + \tensor{B}{^\rho_{\lambda\mu}})
      +
      \tensor{\Gamma}{^{\kappa}_{\sigma\mu}}
      \tensor{B}{^{\sigma}_{\lambda\nu}}
      +
      \tensor{\Gamma}{^{\kappa}_{\lambda\tau}}
      \tensor{B}{^{\tau}_{\mu\nu}}
      +
      \tensor{C}{^\kappa_{\lambda\mu\nu}}
    \end{aligned}
    \right..
  \end{equation}
  Since $\mathbf{B}$ and $\mathbf{C}$ are arbitrary and symmetric in their lower indices, we can fix them one by one in order to cancel out the symmetric parts (first $\mathbf{B}$ in the top row and then $\mathbf{C}$ in the bottom row).
\end{proof}

In the case of the Levi-Civita connection, the normalization result becomes
\begin{cor}
  \label{thm:normalization}
  For any point $m\in(\univ,g)$, there exists a diffeomorphism with trivial linear part $\varphi$ such that, at the point $m$, the Levi-Civita connection of the metric $\pull g$ satisfies
  \[
    \tensor{\Gamma}{^\kappa_{\lambda\mu}}=0
    \quad\text{and}\quad
    \tensor{\Gamma}{^\kappa_{\lambda\mu,\nu}}
    =
    -\frac{1}{3} \big(\tensor{R}{^\kappa_{\lambda\mu\nu}}+\tensor{R}{^\kappa_{\mu\lambda\nu}}\big).
  \]
\end{cor}

\begin{proof}
  The Levi-Civita connection is symmetric, hence
  \[
    \tensor{\bar\Gamma}{^\kappa_{(\lambda\mu)}}=0
    \iff
    \tensor{\bar\Gamma}{^\kappa_{\lambda\mu}}=0.
  \]
  Moreover, at a point where the Christoffel coefficients vanish, the local components of the Riemann tensor are
  $ \tensor{R}{^{\kappa}_{\lambda\mu\nu}} = \tensor{\Gamma}{^{\kappa}_{\nu\lambda,\mu}} - \tensor{\Gamma}{^{\kappa}_{\mu\lambda,\nu}} $.
  One infers, using the symmetries of $\tensor{\Gamma}{^{\kappa}_{\mu\lambda,\nu}}$ in $\mu,\lambda$ that
  \[
    \tensor{R}{^{\kappa}_{\lambda\mu\nu}}
    +
    \tensor{R}{^{\kappa}_{\mu\lambda\nu}}
    =
    \tensor{\Gamma}{^{\kappa}_{\nu\lambda,\mu}}
    +
    \tensor{\Gamma}{^{\kappa}_{\nu\mu,\lambda}}
    -
    2
    \tensor{\Gamma}{^{\kappa}_{\mu\lambda,\nu}}
    =
    \tensor{\Gamma}{^{\kappa}_{(\mu\lambda,\nu)}}
    -
    3
    \tensor{\Gamma}{^{\kappa}_{\mu\lambda,\nu}}.
  \]
  We finally get that
  \[
    \tensor{\Gamma}{^\kappa_{(\lambda\mu,\nu)}}
    =
    0
    \iff
    \tensor{\Gamma}{^\kappa_{\lambda\mu,\nu}}
    =
    -\frac{1}{3}
    \big(\tensor{R}{^\kappa_{\lambda\mu\nu}} + \tensor{R}{^\kappa_{\mu\lambda\nu}}\big),
  \]
  which ends the proof.
\end{proof}

\section{Volume forms and Hodge operator}
\label{sec:volume-and-Hodge}

\subsection{Volume pseudo-form}
\label{subsec:volume-pseudo-form}

The manifold $\univ$ is orientable if and only if there exists a global trivialization $\vartheta$ of the line bundle $\Wedge^{4}\TM$.
Assume now that $\univ$ is oriented. Given a Lorentzian metric $g$ on $\univ$, one can then define the \textsl{Riemannian volume form} $\vol_{g}$ as follows (the dependence on $\vartheta$ is implicit). At a point $m\in\univ$, given the orientation $\vartheta_{m}$ of $\TM[m]$ (seen as a non-zero element of the $1$-dimensional vector space $\Wedge^{4}\TM[m]$), a basis $v \bydef(v_{0},v_{1},v_{2},v_{3})$ of $\TM[m]$ is called \textsl{direct} if
\[
  (v^{0}\wedge v^{1}\wedge v^{2}\wedge v^{3})(\vartheta_{m}) > 0,
\]
where $(v^{0},v^{1},v^{2},v^{3})$ is the dual basis of $\TM[m]^{*}$. Given a direct orthonormal basis $(v_{0},v_{1},v_{2},v_{3})$ of $\TM[m]$, we set
\[
  \vol_{g}(m) \bydef v^{0}\wedge v^{1}\wedge v^{2}\wedge v^{3}.
\]
Since any other direct orthonormal basis is written $Lv$, where $L$ belongs to the special Lorentz group
\[
  \SO(3,1) \bydef \set{L \in \OO(3,1);\; \det L = 1},
\]
this construction does not depend on the choice of the direct orthonormal basis $v=(v_{0},v_{1},v_{2},v_{3})$. The data $\vol_{g}(m)$ are glued together in order to define a global smooth volume form $\vol_{g}\in\Omega^{4}(\univ)$.

For a local diffeomorphism $\varphi\colon U\to \bar U$, recall from \autoref{subsec:diff-action-tensors} that
\[
  \pull g(\pull v_i,\pull v_j)=g(v_i,v_j)
  \quad\text{and}\quad
  \pull\vartheta(\pull v)=\vartheta(v).
\]
Therefore, if $v$ is a direct orthonormal basis of $\TM[\varphi(m)]$ for $(g,\vartheta)$, then $\pull v$ is a direct orthonormal basis of $\TM[m]$ for $(\pull g,\pull\vartheta)$.
Evaluating the volume forms on $\pull v$, one infers that
\[
  \vol_{\pull g,\pull\vartheta}
  =
  \det(T\varphi)^{-1} (\vol_{g,\vartheta})\circ\varphi^{-1}
  =
  \pull\vol_{g,\vartheta}.
\]
In other words, the volume form $\vol_{g,\vartheta}$ is a covariant of the pair $(g,\vartheta)$ under the group of all diffeomorphisms.

If the orientation $\vartheta$ has been fixed, a slight modification of our arguments shows that $\vol_{g}$ is a covariant of the field $g$ under the group of \emph{orientation-preserving} diffeomorphisms, but that it changes sign when the orientation is reversed (it is a density):
\[
  \vol_{\pull g,\vartheta}
  =
  \abs{\det(T\varphi)}^{-1} (\vol_{g,\vartheta})\circ\varphi^{-1}.
\]

In the present work we consider only the subgroup of orientation-preserving diffeomorphisms of $\univ$, which is common in General Relativity.

\subsection{Hodge star operator}
\label{subsec:Hodge}

We consider an oriented Lorentzian $4$-dimensional manifold $(\univ,g,\vartheta)$ with metric signature $(-+++)$.
We fix the orientation $\vartheta$ and from now on leave it implicit. Denote by $\hodge$ the Hodge star operator. Recall that it is $\mathcal{C}^\infty(\univ)$-linear and defined recursively on decomposable $k$-forms by the normalization $ \hodge1=\vol_{g} $ and by the formula
\[
  i_{\pmb V}(\hodge\omega)
  =
  \hodge(\omega\wedge\pmb{V}^\flat),
\]
where $i_{\pmb{V}}$ is the contraction by a vector field $\pmb{V}$.
Detailing the definition of $\hodge$, one has, for all $1$-forms $\omega^1$, $\omega^2$, $\omega^3$, $\omega^4$:
\begin{equation}
  \label{eq:hodge}
  \begin{aligned}
    \vol_{g}(\omega^{1\sharp}, \omega^{2\sharp}, \omega^{3\sharp}, \omega^{4\sharp})
    =
    \hodge(1)(\omega^{1\sharp}, \omega^{2\sharp}, \omega^{3\sharp}, \omega^{4\sharp})
     & =
    \hodge(\omega^{1}) (\omega^{2\sharp}, \omega^{3\sharp}, \omega^{4\sharp})
    \\
     & =
    \hodge(\omega^{1} \wedge \omega^{2})(\omega^{3\sharp}, \omega^{4\sharp})
    \\
     & =
    \hodge(\omega^{1} \wedge \omega^{2} \wedge \omega^{3})(\omega^{4\sharp})
    \\
     & =
    \hodge(\omega^{1} \wedge \omega^{2} \wedge \omega^{3} \wedge \omega^{4}).
  \end{aligned}
\end{equation}
Since the left-hand side of \eqref{eq:hodge} is covariant in $g$ (see \autoref{subsec:diff-action-tensors}, \autoref{subsec:volume-pseudo-form}), these equations show that the Hodge star operator is also covariant in $g$, in the sense that, for any orientation-preserving local diffeomorphism $\varphi$:
\begin{equation}
  \label{eq:hodge_covariant}
  \ast_{(\pull g)}(\pull\omega)
  =
  \pull(\hodge\omega).
\end{equation}
We admit that in Lorentzian signature $(-+++)$
\begin{equation}
  \label{eq:hodge_dual}
  \hodge(\hodge\omega)=(-1)^{\deg(\omega)+1}\omega.
\end{equation}
This is easily checked in a local orthonormal coframe using \eqref{eq:hodge}.

If $\deg(\omega)=3$, one infers from \eqref{eq:hodge_dual} the following alternative implicit definition of the $1$-form $\hodge\omega$. Let $\pmb{W}\bydef(\hodge\omega)^\sharp$.
Since
\[
  i_{\pmb{W}}\vol_g
  =
  i_{\pmb{W}}(\hodge1)
  =
  \hodge(\pmb{W}^\flat)
  =
  \hodge(\hodge\omega),
\]
the vector field $\pmb{W}$ is defined implicitly by
\begin{equation}
  \label{eq:W}
  i_{\pmb{W}}\vol_g
  =
  \omega.
\end{equation}
For a decomposable $3$-form $\omega=\omega^1\wedge\omega^2\wedge\omega^3$, using the definition of $\pmb{W}$, \eqref{eq:hodge} and \eqref{eq:W}, one has
\[
  g(\pmb{W},\pmb{W})
  =
  \hodge\omega\cdot\pmb{W}
  =
  \vol_{g}(\omega^{1\sharp},\omega^{2\sharp},\omega^{3\sharp},\pmb{W})
  =
  -(\omega^1\wedge\omega^2\wedge\omega^3)(\omega^{1\sharp},\omega^{2\sharp},\omega^{3\sharp}),
\]
and thus
\begin{equation}\label{eq:norm_W}
  g(\pmb{W},\pmb{W})
  =
  -\det(g^{-1}(\omega^I,\omega^J)).
\end{equation}

\subsection{Matter coframe}

Let $\Psi\colon\univ\to\RR^3$ be a submersion and suppose that the subspace $\ker(T\Psi)$ is either timelike or spacelike with respect to $g$. Let $\Psi^1,\Psi^2,\Psi^3$ denote its components in the canonical basis of $\RR^3$. Applying \eqref{eq:hodge} to the decomposable $3$-form $\dd\Psi^1\wedge\dd\Psi^2\wedge\dd\Psi^3$, one gets
\[
  \hodge(\dd\Psi^1\wedge\dd\Psi^2\wedge\dd\Psi^3)(\pmb{V})
  =
  \vol_{g}(\grad^{g}\Psi^1,\grad^{g}\Psi^2,\grad^{g}\Psi^3,\pmb{V}).
\]
Looking at the right-hand side, the covector field $\hodge(\dd\Psi^1\wedge\dd\Psi^2\wedge\dd\Psi^3)$ depends only on the pointwise value of the metric $g$ and on the first-order jet of the function $\Psi$. Furthermore, looking this time at the left-hand side, by \eqref{eq:hodge_covariant}, the covector field $\hodge(\dd\Psi^1\wedge\dd\Psi^2\wedge\dd\Psi^3)$ is covariant of $g$ and $\Psi$.

Notice that
\[
  g^{-1}(\hodge(\dd\Psi^1\wedge\dd\Psi^2\wedge\dd\Psi^3),\dd\Psi^I)
  =
  \vol_{g}(\grad^{g}\Psi^1,\grad^{g}\Psi^2,\grad^{g}\Psi^3,\grad\Psi^I)
  =
  0.
\]
Since $\Psi$ is a submersion, the covector fields $\dd\Psi^1,\dd\Psi^2,\dd\Psi^3$ span a subspace of dimension $3$. Therefore the vector field $(\hodge(\dd\Psi^1\wedge\dd\Psi^2\wedge\dd\Psi^3))^\sharp$ is a generator of $\ker(T\Psi)$. This yields by assumption that this vector field is not isotropic, and hence neither is its dual covector field $\hodge(\dd\Psi^1\wedge\dd\Psi^2\wedge\dd\Psi^3)$.
One can sum up some of the previous considerations in the following lemma.

\begin{lem}\label{lem:hodge}
  Let $\Psi\colon\univ\to\RR^3$ be a submersion and suppose that the subspace $\ker(T\Psi)$ is either timelike or spacelike with respect to $g$.
  Then, the family of covector fields
  \[
    \mF(g,\Psi)
    \bydef
    \left(\hodge(\dd\Psi^{1}\wedge\dd\Psi^{2}\wedge\dd\Psi^3), \dd\Psi^{1}, \dd\Psi^{2}, \dd\Psi^{3}\right),
  \]
  is a coframe of $\univ$. It depends only, and in a covariant way, on the pointwise value of the metric $g$ and on the first-order jet of the function $\Psi$.
\end{lem}

\begin{rem}
  In the framework of \textsl{perfect matter}, following Souriau, we make the further hypothesis that the subspace $\ker T\Psi$ is timelike, \textit{i.e.}, that
  \[
    g(\pmb{W},\pmb{W}) < 0.
  \]
  However, the reasoning implemented here can be carried unchanged if we assume that $\ker(T\Psi)$ is spacelike (this would allow one to study the case of \textsl{tachyons}). By contrast, if $\ker(T\Psi)$ is isotropic (the case of \textsl{photons}), our reasoning breaks, as we would not produce $4$ linearly independent covector fields.
\end{rem}

\begin{rem}
  Whereas the covector fields $\dd\Psi^{I}$ are naturally covariant under the group of general diffeomorphisms, the supplementary covector field $\hodge(\dd\Psi^1\wedge\dd\Psi^2\wedge\dd\Psi^3)$ changes sign when the orientation is reversed. It could be made covariant under the group of general diffeomorphisms by making its dependence on the orientation explicit---it would become a tensor density, see \autoref{subsec:volume-pseudo-form} for more details on the volume pseudo-form---, but we feel that this would hinder the clarity of the exposition.
\end{rem}

\textbf{Data Availability Statement.} Data sharing is not applicable, as no data was collected or processed for the research reported in this article.

\medskip

\noindent \textbf{Declarations.}

\textbf{Authors’ Contributions.} Authors’ names are in alphabetical order. All authors contributed equally to this manuscript.

\textbf{Funding.} No funding have been received.

\textbf{Conflict of interest.} The authors declare no conflict of interest.


\end{document}